\documentclass[preprint]{imsart}

\RequirePackage[OT1]{fontenc}
\RequirePackage{amsthm,amsmath,natbib}
\RequirePackage[colorlinks,citecolor=blue,urlcolor=blue]{hyperref}
\RequirePackage{hypernat}
\usepackage{graphicx,subfig,amssymb}

\startlocaldefs
\numberwithin{equation}{section}
\theoremstyle{plain}

\def\bbR{{\mathbb R}}

\def\E{{\mathbb E}}

\newcommand{\beq}{\begin{equation}} 
\newcommand{\eeq}{\end{equation}} 
\newcommand{\var}{{\rm {Var}}}

\newtheorem{proposition}{Proposition}

\newtheorem{theorem}{Theorem}

\def\what#1{\widehat{#1}}
\def\wtilde#1{\widetilde{#1}}
\def\noi{\noindent}

\endlocaldefs

\begin{document}

\bibliographystyle{imsart-nameyear}

\begin{frontmatter}
\title{Network--wide Statistical Modeling and Prediction of Computer Traffic}
\runtitle{Network Modeling and Prediction}

\begin{aug}
\author{\fnms{Joel} \snm{Vaughan}\thanksref{t1}\ead[label=e1]{rsnation@umich.edu}},
\author{\fnms{Stilian} \snm{Stoev}\thanksref{t1}\ead[label=e2]{sstoev@umich.edu}}
\and
\author{\fnms{George} \snm{Michailidis}\thanksref{t1}
\ead[label=e3]{gmichail@umich.edu}}

\thankstext{t1}{The authors were partially funded by the NSF grant DMS--0806094.}
\runauthor{ Vaughan, Stoev, \& Michailidis}

\affiliation{The University of Michigan, Ann Arbor}

\address{439 West Hall\\
1085 S University\\ 
Ann Arbor, MI, 48109 \\
\printead{e1}\\
\phantom{E-mail:\ }\printead*{e2,e3}}

% \address{Address of the Third author\\
% usually few lines long\\
% usually few lines long\\
% \printead{e3}\\
% \printead{u1}}
\end{aug}

\begin{abstract}
In order to maintain consistent quality of service, computer
network engineers face the task of monitoring the traffic fluctuations
on the individual links making up the network. However, due to
resource constraints and limited access, it is not possible to
directly measure all the links. Starting with a physically
interpretable probabilistic model of network–wide traffic, we
demonstrate how an expensively obtained set of measurements may be
used to develop a network--specific model of the traffic across the
network. This model may then be used in conjunction with easily
obtainable measurements to provide more accurate prediction than is
possible with only the inexpensive measurements. We show that the
model, once learned may be used for the same network for many
different periods of traffic.   Finally, we show an
application of the prediction technique to create relevant control
charts for detection and isolation of shifts in network traffic.
\end{abstract}

% \begin{keyword}[class=AMS]
% \kwd[Primary ]{60K35}
% \kwd{60K35}
% \kwd[; secondary ]{60K35}
% \end{keyword}


\end{frontmatter}

\section{Introduction and Preliminaries}

Computer networks consist of nodes (routers and
switches) connected by physical links (optical or copper wires).
Data 
from one node (called a \emph{source})  to another
(\emph{destination})  is sent over the network on predetermined
paths, or routes. 
We will call the stream of data between a particular
source/destination pair a \emph{flow}.  The so--called
\emph{flow--level traffic} may 
traverse only a single link, if the source and destination nodes are
directly connected, or  several links, if they are not.   Also of
interest is the aggregate data traversing each \emph{link}.  The
traffic on a given link is the sum of the traffic of the various flows
using the link.

Both the \emph{flow--level}  and \emph{link--level} traffic have been
studied in the literature.  Flow level data is expensive to
obtain and process, but provides information directly about the
flows.  Data, especially that involving packet delays from
source to destination, has been used to do something. See some
references.  On the other hand, the link level data is less expensive
to obtain, but provides less information about the underlying flows.
These data have been studied extensively by the field of network
tomography.  

This work examines the problem of predicting the traffic level on an
unobserved link via measurements on a subset of the other links in the
network.  Rather than focusing solely on the inexpensive link level data, we
also employ flow--level  measurements to inform a model that can then
be utilized to make solve the prediction problem.  We demonstrate that
this model, although initially requiring the expensive flow--level
data, may be used with a large range of link level data from the same
network.  Thus, this work is the first to combine both flow and link
level data in an efficient and feasible way.

The remainder of the paper is laid out as follows.  Section 2
discusses the computer network framework  and the link--level
prediction problem in detail.  Section 3 describes the proposed
model that utilizes the flow--level traffic.  Section 4 demonstrates
the robustness of the resulting model, while Section 5
illustrates a potential application of the methodologies described
herein in the case when all links may be observed.  

\section{Network Modeling and Prediction}

Here, we introduce some notation and motivate our modeling framework in the
context of network prediction.

\subsection{Global Traffic Modeling in Computer Networks} \label{s:global_model_intro}

Computer networks consist of collections of nodes (routers and
switches) connected by physical links (optical or copper wires). The networks
may be viewed as connected directed or undirected graphs. Figure \ref{fig:scen7} illustrates,
for example, the topology of the Internet2 backbone network \cite{Internet2}, comprised
of 9 nodes (routers) and 26 unidirectional links.

Let $n$, $L$ and ${\cal J}$ denote the number of nodes, links, and routes,
respectively, of a given network. Typically, every node can serve both as a source and a
destination of traffic and thus there are ${\cal J} = n(n-1)$
different routes corresponding to all ordered $(source,destination)$
pairs. Computer traffic from one node to another is routed over predetermined sets of links called 
paths or routes. These paths are best described in terms of the routing matrix
$A = (a_{\ell j})_{L\times {\cal J}}$, where
$$
a_{\ell j}=\left\{\begin{array}{rl} 
1& \mbox{link $\ell$ used in route $j$} \\
0& \mbox{link $\ell$ not used in route $j$ } 
\end{array} \right.   
\ \ \ 1\le \ell \le L,\ 1\le j\le {\cal J}.
$$
The rows of the matrix $A$ correspond to the $L$ links and the columns to the ${\cal J}$ routes.
In the Internet2 network, for example, the route $j$ from Chicago to Kansas City
involves only one link, and thus the $j$--th column of $A$ has a single '$1$' on corresponding to the 
link connecting the two nodes; similarly, the $k$--hop routes correspond to columns of $A$ with
precisely $k$ $1$'s.
 
\begin{figure}[t!]
\includegraphics[width=.7\textwidth]{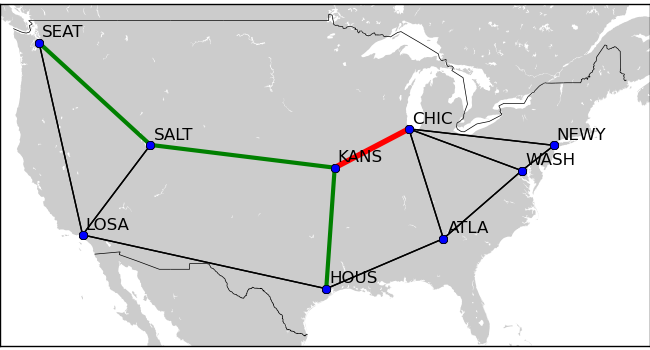}
\caption{The Internet2 topology, with prediction scenario 7 highlighted (see Table \ref{tab:cases}).}
\label{fig:scen7}
\end{figure}

We are interested in the statistical modeling of the traffic on the
{\em entire} network. Let $X(t) =(X_j(t))_{1\le j\le {\cal J}}$  
be the vector of the traffic {\em flows} at time $t$ on all ${\cal J}$ routes, i.e.\ between all
source--destination pairs. That is, $X_j(t),\ t=1,2,\cdots,$ is the number of bytes transmitted 
over route $j$ during the time interval $((t-1)h,th]$, for some fixed $h>0$. Depending on the context,
the time units ($h$) can range from a few milliseconds up to several seconds or even minutes.
Similarly, let $Y(t) = (Y_\ell(t))_{1\le \ell \le L}$ be the vector of the traffic loads at time $t$ over all 
$L$ links. Figures \ref{fig:flows} and \ref{fig:BaOK} illustrate {\em
  flow--}  and {\em link--level} traffic over the Internet2 network.

Assuming that traffic propagates instantaneously through the network,
the load $Y_\ell(t)$ on link $\ell$ at time $t$ equals the cumulative
traffic of all routes 
using this link:
$$
Y_\ell(t)=\sum_{j=1}^{\mathcal{J}} a_{\ell j} X_j(t)
=\sum_{j\in A_\ell} X_j(t),
$$
where $A_\ell \subset \{1,\dots, \mathcal{J}\}$ is the set of routes that
involve link $\ell$. In matrix notation, we obtain
\begin{equation}
\label{routing.eqn1}
Y(t)=A X(t).
\end{equation}
We shall refer to \eqref{routing.eqn1} as to the {\it routing equation}.
In practice, this relationship between the \emph{flow--level} traffic $X(t)$ and
the \emph{link--level} traffic $Y(t)$ is essentially exact, provided that the 
time scale of measurement $h$ is comparable or greater than the maximum {\it
round trip time} (RTT) for packets in the network.  In this paper, we
consider aggregate data over 10 second intervals, a time substantially
greater the the RTT of a packet over the Internet2 backbone.

\begin{figure}
\includegraphics[width=.7\textwidth]{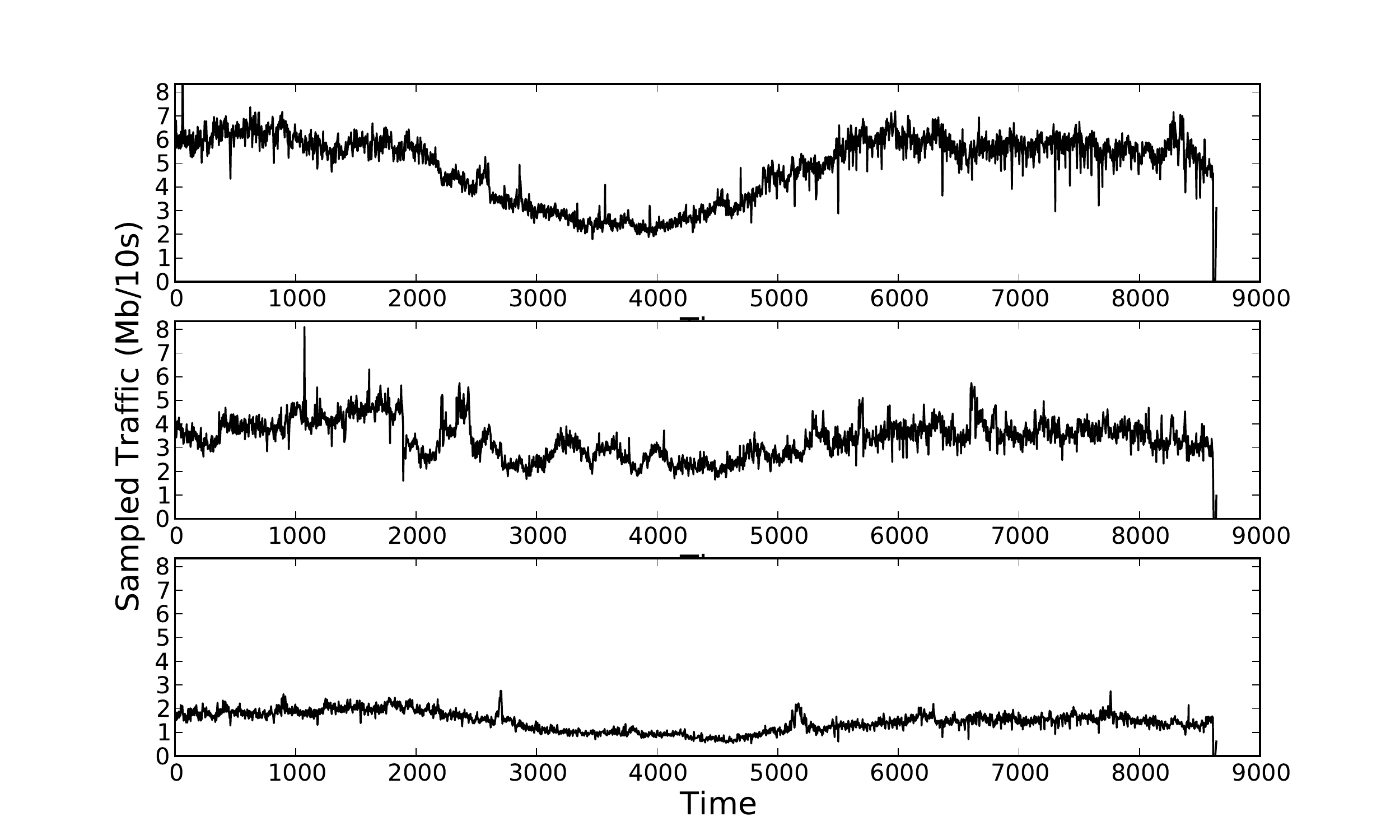}
\caption{Sampled traffic flow ($X_j(t)$) time series for high, medium, and low 
 volume flows on the Internet2 network, recovered from NetFlow data (see 
 Appendix \ref{sec:map}, for technical details).}
\label{fig:flows}
\end{figure}

The statistical behavior of the network traffic traces over {\em individual} links or 
flows has been studied extensively in the past 20 years. An important {\em burstiness} 
phenomenon was observed, which rendered the classical telephone traffic models based on
Poisson and Markov processes inapplicable. Mechanistic and physically interpretable 
models were developed that explain and relate the observed burstiness to the notions of
long--range dependence and self--similarity.  For more details, see
eg \cite{willinger:taqqu:leland:wilson:1995SS},
\cite{taqqu:willinger:sherman:1997}, \cite{park:willinger:2000}, 
\cite{mikosch:resnick:rootzen:stegeman:2002},
\cite{stoev:michailidis:vaughan:2010}, and the references therein.

%%Does this belong?
Under general conditions, limit theorems show that the cumulative
fluctuations of the traffic about its mean can be naturally
modeled with {\em fractional Brownian motions} (fBm). The fBm's are zero mean
Gaussian processes, which have stationary increments and are {\em
self--similar}. For more details, see \cite{willinger:paxson:riedi:taqqu:2003-livre} 
or Section \ref{s:adet} below. Other limit regimes are also possible,
which lead to stable infinite variance models. These models,
however, are less robust to network configurations than the fractional Brownian motion and are
encountered less frequently in practice (see eg \cite{mikosch:samorodnitsky:2007}).

Here, our focus is not on the temporal dependence of traffic,
which has been studied extensively.  Rather, our goal is to model the global
structural relationship between all links on the network at a fixed
instant of time, $t_0$.  Motivated by existing single--flow models, 
we shall suppose that the vector $  X(t)$ representing the traffic on all flows of
the network is multivariate normal: 
\begin{equation}
  X(t) \sim N ( \mu_X(t), \Sigma_X(t)).
\label{eqn:xdist}
\end{equation}
As discussed in \cite{stoev:michailidis:vaughan:2010}, the
link--level, as well as flow--level traffic can be well--modeled by
using multivariate long--range dependent Gaussian processes. The routing
equation \eqref{routing.eqn1} yields: 
\begin{equation}
   Y(t) \sim N (A  \mu_X(t), A\Sigma_X(t)A^t).
 \label{eqn:ydist}
\end{equation}

A first natural application and in fact our motivation to develop
global network traffic models comes from the {\em
 network prediction} problem. The performance of our models
 will be explored and illustrated in this context.

\medskip
\noindent {\sc Network Prediction (Kriging):} \begin{it} \label{pr:kriging}
Observed are the traffic traces on a subset of links ${\cal O} \subset \{1,\cdots,L\}$:
\begin{equation}\label{e:D}
 {\cal D}(t_0,m) := \{ Y_\ell(t),\ t_0-m+1\le t\le t_0,\ \ell \in {\cal O} \}.
\end{equation}
 over the time window $t_0-m+1 \le t \le t_0$ of size $m$.

 Obtain estimators $\what Y_\ell (t_0)$ for the traffic $Y_\ell(t_0)$,
 for all unobserved links $\ell \in {\cal U} := \{1,\cdots, L\}
 \setminus {\cal O}$ at time $t_0$ in terms of the data ${\cal D}(t_0,m)$. 
\end{it}

\smallskip
In view of the Gaussianity of our global traffic model, we focus on linear predictors.
In the next section, we start by adapting the 
classical {\em ordinary kriging} methodology from Geostatistics to the
networking context.  This is perhaps the best that one 
can do given no prior statistical information about the network other
than its topology and routing.  Even in such a limited 
setting it turns out that one can sometimes obtain useful predictors.

In many $({\cal O},{\cal U})$--scenarios, however, {\em ordinary
 network kriging} does not provide helpful results.  The main problem is the lack of information 
about the traffic means (and to a lesser degree covariances) on the unobserved
links.  This issue can only be resolved by incorporating additional,
network--specific information about the traffic.
We do so, in Section \ref{sec:network-model}, where by exploring (off--line)
large NetFlow data sets, we recover estimates of all flows $X_j(t)$'s
in the network. These data are used to build a flexible
network--specific model that can be estimated {\em on--line} from
link--level measurements. The resulting model will be shown to 
substantially outperform the {\em ordinary network kriging} methodology.

\subsection{General Theory}\label{sec:general}

If the means and covariances are known, then the network prediction problem in the
previous section becomes the {\em simple kriging} problem from Spatial Statistics. 
The best linear unbiased predictor (BLUP) for $Y_u=Y_u(t)$ in terms of
$Y_o=Y_o(t)$ is then given by: 
\begin{equation}
\widehat{Y}_u= {\mu}_u+\Sigma_{uo}\Sigma^{-1}_{oo}( {Y}_o- {\mu}_o),
\label{eqn:krig-est}
\end{equation} 
where 
$$
Y = \left(\begin{array}{l} Y_o \\ Y_u \end{array} \right),\ \mu_Y
=  \left(\begin{array}{l} \mu_o \\ \mu_u \end{array} \right),\ \mbox{ and }\ 
\Sigma_{Y} = \left(\begin{array}{ll} 
\Sigma_{oo} & \Sigma_{ou}\\
 \Sigma_{uo} & \Sigma_{uu}
\end{array} \right).
$$
Here $\E Y(t) = \mu_Y(t)$ and $\Sigma_Y=\Sigma_{Y}(t)$ are the mean
and the covariance matrix of the vector $Y(t)$, which is partitioned into an 
\emph{observed} and \emph{unobserved} components $Y_o(t)$ and $Y_u(t)$, respectively. 
For simplicity, we omit the argument `$t$' in \eqref{eqn:krig-est},
but it is implicitly present in all quantities therein. 
The (conditional) covariance matrix of prediction errors $  Y_u - \what Y_u$ is given by:
\begin{equation}
\E {\Big(} ( {Y}_u-\widehat{Y}_u)( {Y}_u-\widehat{Y}_u)^t | Y_o {\Big)}  
= \Sigma_{uu}-\Sigma_{uo}\Sigma^{-1}_{oo}\Sigma_{ou}.
\label{eqn:krig-mse}
\end{equation}

The performance of the simple kriging predictor is illustrated in
Figure \ref{fig:BaOK} (top), and Table \ref{tab:okrig} ($2$nd
column). In these cases, $\mu_Y(t)$ and $\Sigma_Y(t)$ are estimated
from moving windows of past observations, where data on all links is
available (see Appendix \ref{ap:okrig}). In the context of our network
prediction problem, however, not all links are observed, the means and
covariances are unknown, and the {\em simple kriging} methodology is
not practical. Nevertheless, it provides a theoretically optimal
benchmark that we will be used to evaluate all methods developed in the sequel.

\begin{figure}[t!]
\includegraphics[width=\textwidth, height= 3 in]{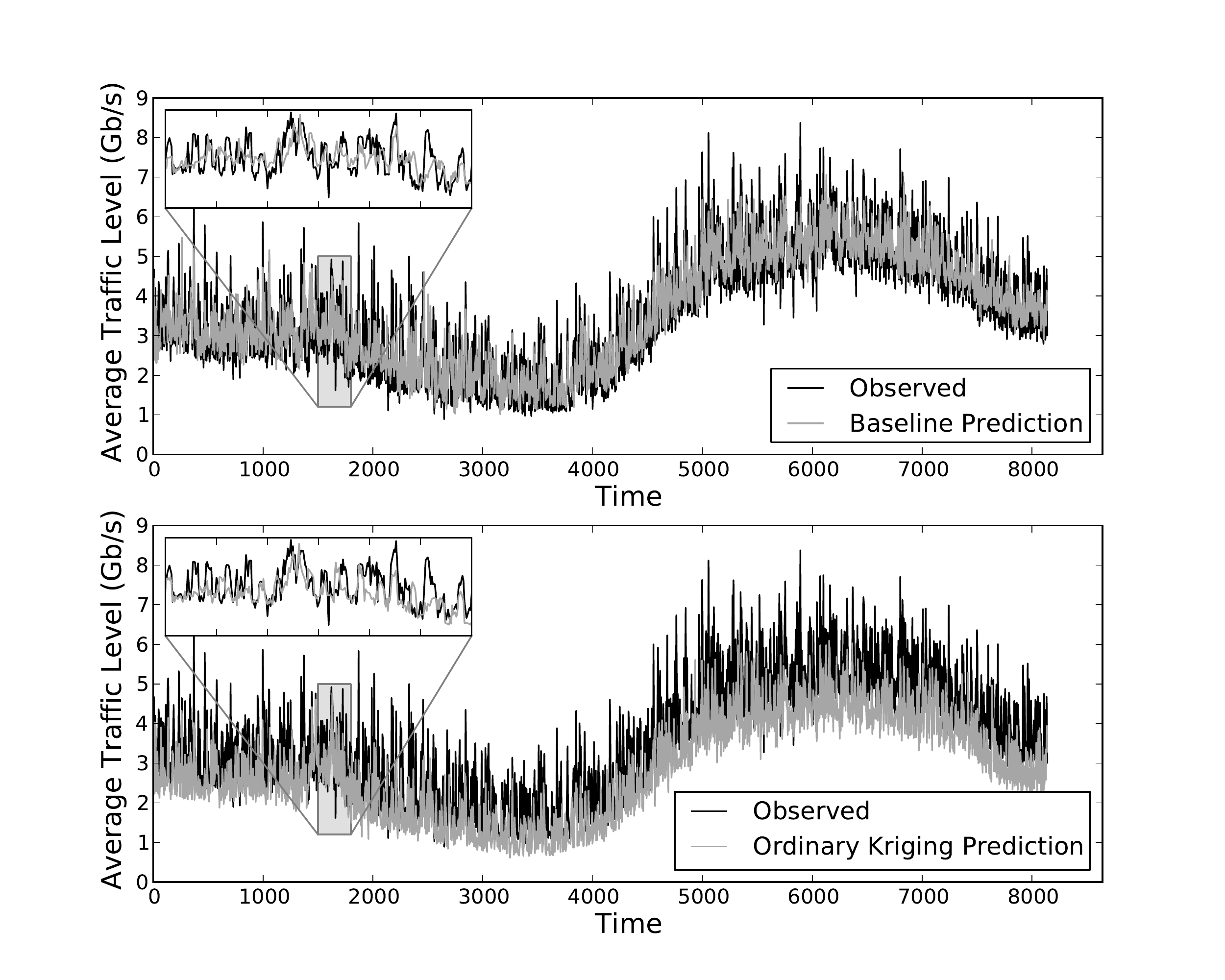}
\caption{{\it Top:} the \emph{standard kriging} estimator
  based on past empirical means and covariances. This
  \emph{baseline} estimator is not available in practice when only a
  subset of links is observed. {\it Bottom:} the \emph{ordinary kriging} 
  estimator. Both plots represent Scenario 7: the prediction of link 13 via 
  links 3, 9, and 12, depicted in Figure \ref{fig:scen7}.}
\label{fig:BaOK}
\end{figure}

The relative mean squared errors (ReMSE) of the predictors reported in Table 
\ref{tab:okrig}, and in the sequel, are computed as follows:
\begin{equation}\label{e:ReMSE}
 {\rm ReMSE}(\what Y) = \sum_{t=1}^T \| \what Y(t) - Y(t)\|^2 / \sum_{t=1}^T \|Y(t)\|^2.
\end{equation}
Here $\what Y(t)$ is a predictor for the true value $Y(t)$ and
$\|\cdot\|$ stands for the Euclidean norm.  
The ReMSE's quantify empirically the prediction error relative to the
energy of the true signal $Y(t)$, over the 
duration $T$.  In a controlled setting where $Y(t)$ is available, the
ReMSE's allow us to objectively compare the performance of various estimators.

A first practical solution to the network prediction problem may be
obtained by using the {\em ordinary kriging} methodology. 
Namely, suppose that we have no prior information about the
statistical behavior of the traffic over the network, but the 
routing matrix is available.  Then reasonable working assumptions are that
$\E Y(t) = \mu(t)   \vec{1}$ and 
$\Sigma_X(t) = \sigma_X^2(t) \mathbb I_{{\cal J}}$, where $\mu(t)\in \mathbb R$ and
$\sigma_X(t)>0$ are unknown constant parameters, which vary slowly with time $t$.
That is, the traffic means of all links are equal to $\mu(t)$, and the covariance
matrix of the flow--level traffic is scalar. Using these assumptions
together with the routing equation \eqref{routing.eqn1}, one can obtain the 
best linear unbiased predictor of $Y_u(t)$ (see Appendix \ref{ap:okrig} and 
\cite{cressie:1993} for more details). 

\begin{table}
\begin{center}
\begin{tabular}{r|rr|rrrr}
  \hline
Scenario& Baseline  & 2009-02-19 & 2009-02-18 & 2009-02-20 & 2009-02-26 & 2009-03-12 \\ 
  \hline
1 & 0.0305 & 0.4052 & 0.4212 & 0.4250 & 0.4383 & 0.3708\\
2 & 0.0287 & 0.1266 & 0.1194 & 0.1292 & 0.1072 & 0.1068\\
3 & 0.0288 & 0.3279 & 0.3315 & 0.3432 & 0.3151 & 0.3368\\
4 & 0.0290 & 0.4193 & 0.4342 & 0.4391 & 0.3922 & 0.4460\\
\hline
5 & 0.0314 & 0.1209 & 0.0807 & 0.0958 & 0.0962 & 0.1122\\
6 & 0.0285 & 1.0241 & 0.8644 & 0.9323 & 0.8897 & 0.6881\\
7 & 0.0262 & 0.1129 & 0.1225 & 0.1241 & 0.1330 & 0.1435\\
8 & 0.0216 & 0.0614 & 0.0585 & 0.0628 & 0.0805 & 0.0880\\
9 & 0.0242 & 0.1079 & 0.1011 & 0.1059 & 0.1463 & 0.1294\\
\hline
10 & 0.0766 & 12.6471 & 10.5816 & 10.2204 & 10.6031 & 10.1767\\
11 & 0.0727 & 0.8423 & 0.7394 & 0.6346 & 0.6182 & 0.7268\\
12 & 0.0723 & 0.2649 & 0.2338 & 0.2274 & 0.2132 & 0.2486\\
   \hline
\end{tabular}
\end{center}
\caption{{\em Columns 2 and 3}: ReMSE's (see \eqref{e:ReMSE})
of the baseline (simple kriging) and ordinary kriging estimators for
February 19, 2009, in 12 prediction scenarios (Tables  \ref{tab:I2links} 
and \ref{tab:cases}). {\em Columns 4 to 7:} ReMSE's for the ordinary 
kriging estimators over 4 additional days in each scenario.} 
\label{tab:okrig}

\end{table}

Figure \ref{fig:BaOK} (bottom) and Table \ref{tab:okrig} demonstrate the
performance of this ordinary kriging methodology.   
It is certainly uniformly inferior to the benchmark estimator in Figure \ref{fig:BaOK}
(top) and Table \ref{tab:okrig} (column 2), which involves data on the unobserved
links.  Nevertheless, ordinary kriging may be useful in cases when no
structural information about the network is available.

In the following section, we will develop a model that improves
substantially upon the ordinary kriging methodology.  This can be
done, however, only by utilizing further information about the network.

\section{Network Specific Modeling via NetFlow Data}
 \label{sec:network-model}

\subsection{Modeling Traffic Means} \label{sec:means}

Direct measurements of the flow--level traffic $  X (t)$ are very expensive 
to obtain because this would involve examining the entire traffic load of the network, 
i.e.\ storing and then processing 95--170 Gigabytes of data per day. Modern 
routers, however, implement a mechanism called ``NetFlow'', which allow for random (Bernoulli)
sampling  of the flow of traversing packets. The routers  store 
important information such as the ports, source and destination IP
addresses, etc. from a sample of packet headers.
Even though in fast backbone networks (eg Internet2) the practical sampling rates are eg
1 out of 100 packets, the NetFlow mechanism provides unique information about the traffic loads 
in the network. Using a careful mapping procedure, we assigned the sampled packets 
to one of the 72 source/destination flows.  We thus constructed an estimate $\{\wtilde X(t)\} 
\approx \{  X (t)\}$ of the flow--level traffic.  Unfortunately, this method is computationally
expensive to implement, which makes it impractical to use repeatedly, and it is difficult 
to apply in an on--line fashion.  Therefore, the information derived from NetFlow can only be 
viewed as auxiliary data in the context of network prediction.  We
shall use this information to
build a flexible network--specific model that can be estimated on--line.

\begin{figure}[t!]
\subfloat[Windowed
Means \label{fig:hm}]{\includegraphics[width=.49\textwidth]{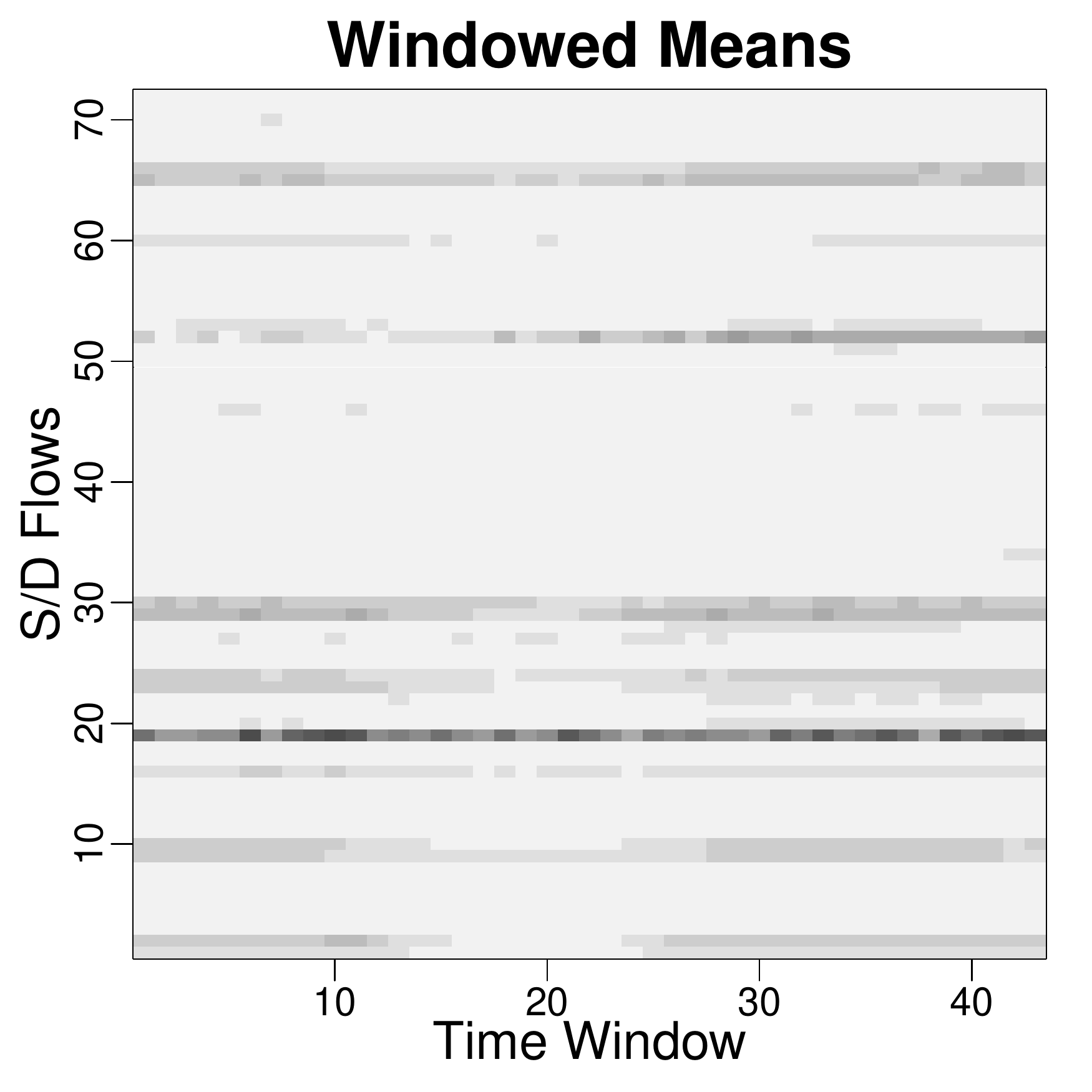}}
\subfloat[Energy \label{fig:rolep}]{\includegraphics[width=.46\textwidth]{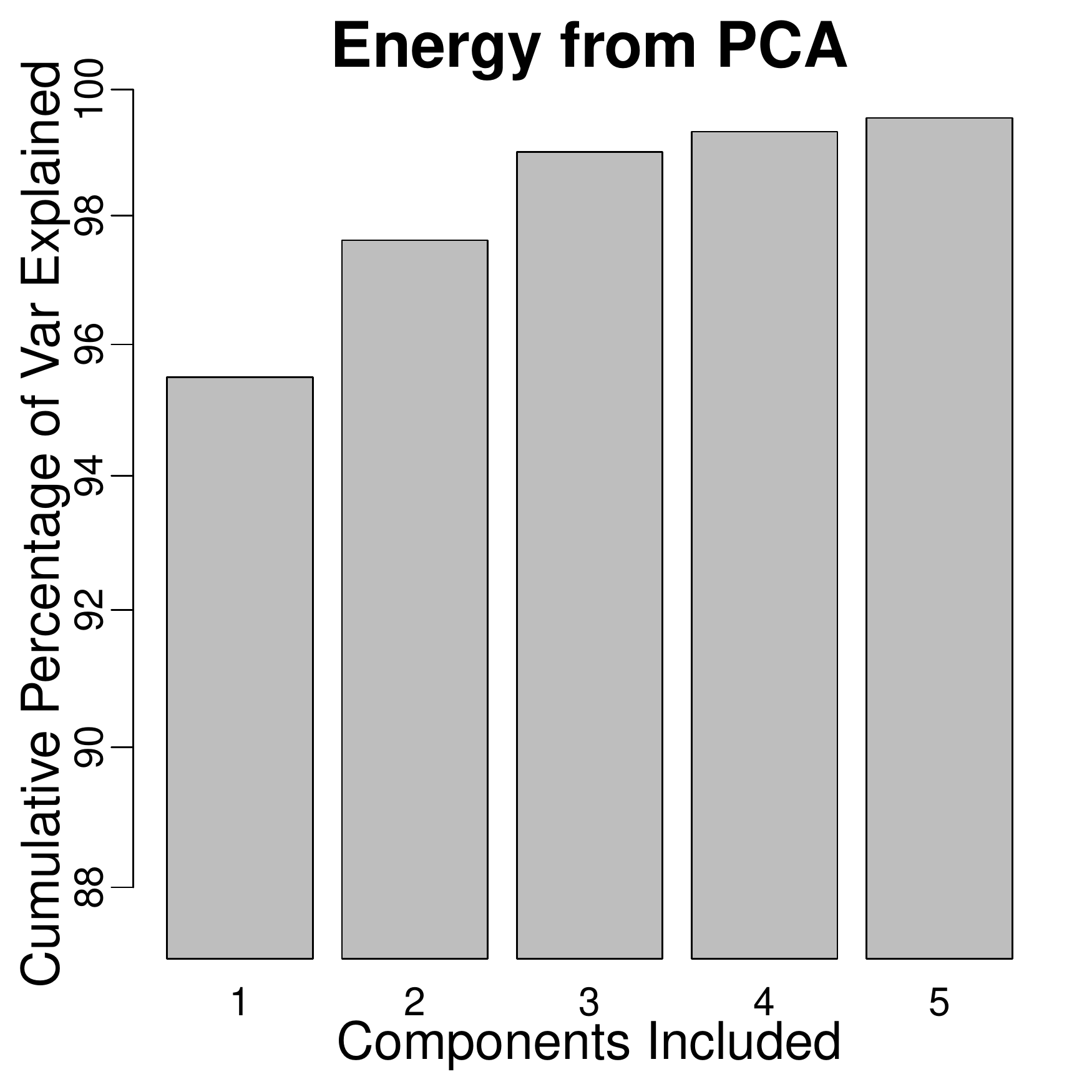}}
\caption{{\it Left:} The columns correspond to local sample means over
  consecutive windows of 2000 seconds for each of the ${\cal J}=72$ flows.
  Darker shades indicate higher values.  The data were reconstructed from NetFlow
  measurements of the Internet2 network for Feb 19, 2009.
{\it Right:} cumulative energy captured by the $F$ matrix for increasing values 
 of $p$ (see \eqref{eq:Fmod} and Proposition \ref{p:PCA} below.)}
\end{figure}

Figure \ref{fig:hm} illustrates the  local means of the
source/destination flows as a function of time, where the data  
were derived from an extensive analysis of NetFlow measurements.  This
suggests that a linear model for $  \mu_X(t)$  
with a few constant factors can capture much of the variability in
the local means. We therefore posit the model 
\begin{equation}
  \mu_X(t) = F  \beta(t),
\label{eq:Fmod}
\end{equation}  
where $F$ is a suitably chosen $\mathcal{J}\times p$ matrix and $\beta
= \beta(t) \in \mathbb R^p$ is a parameter. Observe that
$$
 \mu_Y=A\mu_X=AF\beta\ \ \mbox{ and also } \ \ \mu_{Y_{\rm o}} = A_oF \beta.
$$
Provided $p$ equals ${\rm rank}(A_oF)$ the parameter $\beta$ can be successfully 
estimated by using linear regression from the available data on the {\em observed links}.
We will see that this essentially means that $p$ is no greater than the number 
of observed links $|{\cal O}|$.

Now, our goal is, given $p$, to choose $F$ {\em optimally} so that $F\beta$ can approximate best $\mu_X$
with a suitable $\beta$.  Consider the sample $\wtilde X(t),\ 1\le
t\le n$, of the flow--level data derived from the NetFlow mapping, 
where $n = w\times n_w$. Partition the data into $n_w$ windows of size $w$, and let
$$
\overline X(k) = \frac{1}{w} \sum_{i=1}^w \wtilde X((k-1)w+i)\
$$
$(1\le k \le n),$ be the sample mean the $\wtilde X(t)$'s in the $k$--th window.

Consider first the set of $n_w$ points $\{\overline X(k),\ 1\le k\le n_w\}$
in $\bbR^L$ and observe that the model in \eqref{eq:Fmod} postulates that $\mu_X$
belongs to ${\rm range}(F)$ (the linear space spanned by the columns of $F$).  Thus, 
given the $\overline X(k)$'s, a {\it least squares optimal}  
choice of $F$ corresponds to minimizing the sum of the squared
distances from the $\overline X(k)$'s to  
the $p-$dimensional subspace $W:= {\rm range}(F).$  That is, we want to find 
$$
 W^* = \mathop{{\rm Argmin}}_{W \le \bbR^L,\ {\rm dim}(W) = p} 
 \sum_{k=1}^{n_w} \|\overline X(k) - P_W(\overline X(k))\|^2,
$$
where $P_W$ denotes the orthogonal projection onto the subspace $W$.
The following result shows that this problem has a simple solution,
which corresponds precisely to performing {\it principal component analysis} 
(PCA) on a certain matrix.

\begin{proposition}\label{p:PCA} Let $x(k)\in \bbR^m,\ 1\le k\le n$. Consider
the positive semidefinite  $m\times m$ matrix $B=\sum_{k=1}^{n}  x(k) x(k)^t$ and let 
$
B = \sum_{j=1}^m \lambda_j   b_j  {b}_j^t,
$ 
be its spectral decomposition, where $  b_j,\ 1\le j\le m$ are orthonormal
and $\lambda_1\geq\lambda_2\geq\cdots\geq\lambda_m\geq 0$.  

Set $W^* = {\rm span}\{ {b}_1,\dots, {b}_p\}$, $1\le p\le m$.
Then, for all $W\le \bbR^m$ with ${\rm dim}(W) = p$, 
we have that
\begin{equation}\label{e:p:PCA}
 \sum_{k=1}^m \|x(k) - P_{W^*}(x(k))\|^2 \equiv \sum_{j=p+1}^m
 \lambda_j \le  \sum_{k=1}^n \|x(k) - P_W(x(k))\|^2, 
\end{equation}
where $P_W$ denotes the orthogonal projection onto the subspace $W$.
\end{proposition}
\noindent
The proof is given in Appendix\ref{sec:proofs}. This result implies that the least square optimal 
choice of the matrix $F$ is $ F = ( {b}_1,\cdots, {b}_p)_{{\cal J}\times p}$, where 
the $  b_i$'s, $1\le i\le p$ are the eigenvectors of the largest $p$ eigenvalues of 
the matrix 
$
 B = \sum_{k =1}^{n_w} \overline X(k) \overline X(k)^t.
$
Figure \ref{fig:rolep} shows that just a few PCA factors $p$ are enough to
capture a large percentage of the local variability of the mean vectors $\overline X(t)$.
Figure \ref{fig:MaFM} (top) illustrates the prediction performance of this model when the covariance 
matrix is known, but the {\em unobserved} means are estimated from the model.

\begin{figure}[t!]
\includegraphics[width=\textwidth,height=3in]{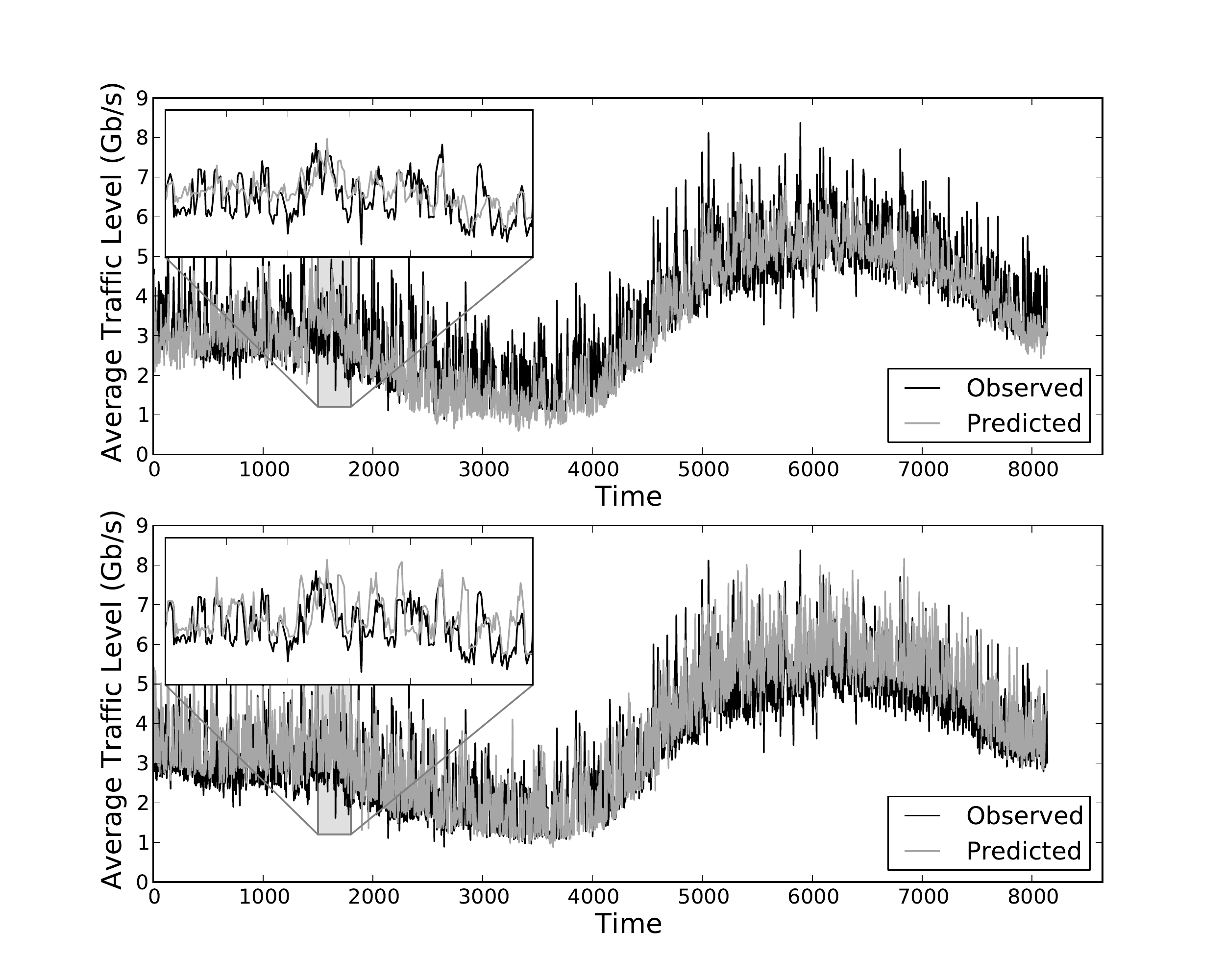}
\caption{{\it Top:} Prediction in Scenario 7 (Table \ref{tab:cases} and Figure \ref{fig:scen7})
 using the PCA--mean model (with $p=2$) and the sample covariance matrix. In reality, the sample
 covariances for unobserved links are not available, and this plot merely 
 illustrates that the model \eqref{eq:Fmod} successfully captures the structure of 
 the means. See also Figure \ref{fig:BaOK}. {\it Bottom:} Prediction in Scenario 7 using the 
 complete mean--variance model with $p=2$. } 
\label{fig:MaFM}
\end{figure}

\subsection{Modeling the Covariances}

The physical nature of network protocols, the mechanisms of
transmission, and the user behavior imply strong relationship between
the means and the variances of traffic traces.  This relationship was
shown to be ubiquitous over different types of computer networks.  In the 
field of network tomography, for example, the mean--variance models
have been successfully used to resolve challenging identifiability 
questions (See eg
\cite{vardi:1996,xi:michailidis:nair:2006,lawrence:michailidis:nair:2006,singhal:michailidis:2007}).  
In our context, 
we also encountered a strong relationship between the means and the
variances of {\em traffic flows}.  More precisely, by exploring the 
sample means $\overline X_j(t)$ and standard errors $S_j(t)$,
calculated over a window of traffic data, we observed that 
\begin{equation}\label{e:Sj_vs_Xj}
 S_j(t) \approx C (\overline X_j(t))^\gamma,\  \ 1\le j \le {\cal  J}
\end{equation}
with $\gamma \approx 3/4$.  Namely, the standard error of a
source--destination flow $X_j(t)$ is proportional to a 
power of its mean.  

We estimated $\gamma$ (as a function of $t$) by performing  
log--linear regression of $S_j(t)$ versus $\overline X_j(t)$ {\em over } 
$j,\ 1\le j\le {\cal J}$. The resulting estimates remained approximately constant in $t$ and close to 
$3/4$ regardless of the time window used.  The power--law relationship is remarkably 
consistent in time and the regression diagnostics $R^2\approx 80\%$ indicate strong agreement with
the model.  For more details, see Figure \ref{fig:gamest} and the
discussion in Section \ref{sec:cal}.

\cite{singhal:michailidis:2007} have shown that the components of the vector 
of $  X(t)$ are, in practice, uncorrelated or at most weakly
correlated.  The principal exception are pairs of {\em forward} and {\em reverse} flows, eg the
Chicago--Los Angeles and Los Angeles--Chicago.  This correlation 
is arguably due to the feedback mechanism built in the TCP protocol.
Our experience with NetFlow on Internet2 (eg Fig.\ 2 in \cite{stoev:michailidis:vaughan:2010}) 
and limited NS2--simulations (\cite{NS2}) confirm that this correlation is negligible 
at our time scales of interest, provided that the network is not congested. 
The study of heavy traffic scenarios beyond the operating
characteristics of the network is interesting but it is  
outside the scope of the present work.  Therefore, in this paper we
shall model $\Sigma_X(t)$ as a diagonal matrix. 

In view of this analysis, we shall impose the following structure on
the flow covariances: 
\begin{equation}
 \Sigma_X=\sigma^2\mbox{diag}(|F\beta|^{2\gamma} ),
\label{eq:Fcov}
\end{equation}
where $|a|^{2\gamma}$ denotes $(|a_i|^{2\gamma})_{i=1}^{\cal J},$ for
$a=(a_i)_{i=1}^{\cal J} \in\mathbb R^{\cal J}$, and where $\mu_X = F\beta$.

We will show below that the parameter $\sigma$ can be estimated
on--line from link--level data ($Y_\ell$'s). On the other hand, $\gamma$  
is a structural parameter, obtained from the off--line analysis of
NetFlow data.  (See Section \ref{sec:cal} for more details.)
% Extensive experiments with a number of NetFlow data sets over many days indicate 
% that the value of $\gamma$ is rather consistent (see
% \cite{vaughan:stoev:michailidis:2010_tech_rep}). 
% \cite{singhal:michailidis:2007} suggest that the compound nature of
% network traffic may be  used to explain and model further this
% ubiquitous mean--variance relationship.    

\subsection{The Joint Model}
\label{s:joint_model}

Combining the mean and covariance models from the previous two
sections, we obtain the following complete model:
\begin{equation}
 Y(t) = AF\beta + \sigma A {\rm diag}(|F\beta|^\gamma)   Z(t),
\label{eq:CompMod}
\end{equation}
where $  Z(t)\sim{\mathcal N}(  0, I_{\cal J})$ is a standard
normal vector in $\mathbb R^{\cal J}$ 
and where $\beta\in \mathbb R^p$ and $\sigma>0$ are unknown
parameters. In this section, we will show how this model can be 
estimated from {\em on--line} measurements on a limited set of 
observed links ${\cal O}$. We will also establish asymptotic properties 
of the proposed estimators.

In the framework of the Network Prediction Problem (see Section \ref{s:global_model_intro}), 
we obtain
$$
 \overline Y_o(t_0) = A_o F \beta + \epsilon_{\overline Y_o}(t_0),
$$
where
\begin{equation}\label{e:Yo-bar}
 \overline Y_o(t_0) = \frac{1}{m} \sum_{k=0}^{m-1} Y_o(t_0-k).
\end{equation}
To establish the covariance structure of the noise $\epsilon_{\overline Y_o}$, we
introduce the mild assumption that the flow--level traffic is stationary 
(in practice, traffic is locally stationary on the time scales of interest) 
and its temporal correlation structure is the same across all routes. Namely, that 
$
 {\rm Corr}(X_j(t),X_j(t+i)) = \rho(i),\ 1\le i\le m-1,\ (1\le j\le {\cal J}).
$
This yields 
$$
 {\rm Corr}(Y_\ell(t+i), Y_\ell(t)) = \rho(i),\ 1\le i\le m-1,\ \ \mbox{ for all }1\le \ell \le L,
$$
and consequently 
\begin{equation}\label{e:epsilon-Y-bar}
 \epsilon_{\overline Y_o}(t_0) \sim {\cal N}(0, \sigma_m^2 A_o {\rm diag}(|F\beta|^{2\gamma})A_o^t),
\end{equation}
where 
\begin{equation}\label{e:sigma-m}
 \sigma_m^2 = \frac{\sigma^2}{m} {\Big(} 1+2\sum_{i=1}^{m-1}(1-i/m) \rho(i) {\Big)}.
\end{equation}

The structure of the noise variance suggests a natural iterated
generalized least squares (iGLS) scheme for the estimation of $\beta$.  

\medskip
\noi {\bf Algorithm:} {\em (Iterated GLS)}

\smallskip
 {\it (i)} Set $\what \beta_1 = [(A_oF)^t A_oF]^{-1} (A_oF)^t
 \overline Y_o(t_0)$ to be the OLS (ordinary least squares) estimate
 of $\beta$ and 
let $k:=1$.

\smallskip
 {\it (ii)} Set 
\begin{equation}\label{e:beta-k}
\what \beta_{k+1} = [(A_oF)^t G(\what \beta_k) A_oF]^{-1} (A_oF)^t
G(\what \beta_k) \overline Y_o(t_0), 
\end{equation}
where 
\begin{equation}\label{e:G(beta)}
 G(\beta) := [ A_o {\rm diag}(|F \beta|^{2\gamma})A_o^t]^{-1}.
\end{equation}

\smallskip
{\it (iii)} Set $k:=k+1$ and repeat step {\it (ii)}.  Iterate until
$\|\what \beta_{k+1} -\what \beta_k\|$ falls below 
 a certain ``convergence'' threshold.

\medskip
Observe that the temporal correlation structure does not need to be
estimated here since it appears only in the scalar coefficient $\sigma_m^2$ 
of the noise variance, which cancels in \eqref{e:beta-k}. The above iGLS scheme requires 
that the matrices involved in steps {\it (i)} and {\it (ii)} be invertible.  The
following result ensures that this is indeed the case under mild 
natural conditions. 

\begin{proposition}\label{p:GLS} Suppose that $F\beta>\vec 0$ and let
  $A_o$ be of full row--rank.  Then: 

{\it (i)} The inverse $G(\beta)$ in \eqref{e:G(beta)} exists, for all $\gamma>0$.

{\it (ii)} If $A_o F$ is of full column--rank, then the inverses
\begin{equation}\label{e:Sigma_GLS}
 \Sigma_{GLS}(\beta):= [(A_oF)^t G(\beta) A_oF]^{-1},
\end{equation}
and $[(A_oF)^tA_oF]^{-1}$ exist and are positive definite.
\end{proposition}

\noi The proof is given in the Appendix and the assumptions are
discussed in the remarks below. 

Now, if the true parameter $\beta$ is known, then as shown in Chapter
3 of \cite{cressie:1993},  
the {\em minimum variance unbiased predictor} (MVUP) of $Y_u (t_0)$
given the data ${\cal D}$ is the standard  
kriging estimate 
$$
 \wtilde Y_u := A_u F\beta + \Sigma_{uo}(\beta)
 \Sigma_{oo}^{-1}(\beta) ( Y_o (t_0) - A_oF\beta) =: f(\beta,Y_o). 
$$
By \eqref{eqn:krig-mse} and \eqref{eq:CompMod}, the covariance matrix
of the prediction error  
$Y_u-\wtilde Y_u$ equals:
\begin{equation}\label{e:MSPE}
 {\rm m.s.e.}(\wtilde Y_u) = \sigma^2 {\Big(} \Sigma_{uu}(\beta) -
 \Sigma_{uo}(\beta) \Sigma_{oo}^{-1}(\beta) 
  \Sigma_{ou}(\beta) {\Big)},
\end{equation}
where 
\begin{equation}\label{e:Sigma-beta}
 \Sigma(\beta) = A {\rm diag}(|F\beta|^{2\gamma}) A^t = \left( \begin{array}{ll}
  \Sigma_{oo}(\beta) & \Sigma_{ou}(\beta)\\
   \Sigma_{uo}(\beta) & \Sigma_{uu} (\beta)
  \end{array}\right).
\end{equation}

In practice, consider the {\em plug--in predictor} $\what Y_u$,
where $\what \beta$ is some estimate  
of $\beta$.  Namely, 
$$
 \what Y_u := f(\what \beta, Y_o).
$$ 
Since the function $f$ is continuous, the consistency of $\what \beta$
would then imply that the plug--in predictor is a consistent estimator of the 
MVUP $\wtilde Y_u$.  The following result establishes the {\em strong}  
consistency of the iterated GLS estimators $\what \beta_k$'s, even in
the presence of long--range dependence. It also shows 
that the $\what \beta_k$'s are asymptotically equivalent to the
(unavailable) GLS estimator $\what \beta_{GLS}$, provided  
$k\ge 2$. 

\begin{theorem}\label{t:asymptotics}
Suppose that $A_o$ and $A_oF$ are of full row and column ranks, respectively, and let 
$F\beta>\vec 0$. Then:

 {\it (i)} For all $k\ge 1,$ we have $F\what \beta_k > \vec 0, a.s.$
 as $m\to\infty$. Hence, 
the estimates $\what \beta_k,\ k\ge 1$ are well--defined, almost surely, as $m\to\infty$.
 
{\it (ii)} If $\rho(\tau) \to 0,$ as $\tau\to\infty$, then for any
fixed $k\ge 1$, we have 
\begin{equation}\label{e:consistency}
 \what \beta_k \stackrel{a.s.}{\longrightarrow} \beta,\ \ \mbox{ as }m\to\infty. 
\end{equation}

{\it (iii)} For all $k\ge 2$, we have that,
$$
 \what \beta_k - \what \beta_{GLS} = o_P(\sigma_m),\ \mbox{ as } m\to\infty,
$$
where $\what \beta_{GLS}$ is the GLS estimate of $\beta$ 
in the model \eqref{eq:CompMod}, and $\sigma_m^2$ is given in \eqref{e:sigma-m}.
Moreover, ${\rm Var}(\what \beta_{GLS}) = \sigma_m^2
\Sigma_{GLS}(\beta)$ with $\Sigma_{GLS}$ as in 
\eqref{e:Sigma_GLS}.  
\end{theorem}

\noi The proof is given in Appendix \ref{sec:proofs}. The asymptotic variance 
of $\what \beta_k,\ k\ge 2$ is discussed in the remarks below. 

As mentioned above, the scale parameter $\sigma$ of the covariance
structure in \eqref{eq:CompMod} is not involved in  
the formula for the predictors $\what Y_u = f(\what\beta_k, Y_o)$ (see
eg \eqref{e:beta-k}).  The parameter $\sigma$ is 
involved, however, in the expression of the prediction error
\eqref{e:MSPE}. Therefore, to gauge the accuracy of prediction, and to  
be able to use our estimators for detection of anomalies (see Section
\ref{s:adet} below), one needs an estimate of $\sigma$. 

As in the case of the ordinary kriging estimator (see
\eqref{e:sigma_X(t)} below), a natural estimate of $\sigma$ is  
obtained as follows:
\begin{equation}\label{e:sigma-hat}
\what\sigma^2 := { \mbox{vec}(\what \Sigma_{Y_o})^t \mbox{vec}(
  \Sigma_{oo}(\what\beta))/ \mbox{vec}(\Sigma_{oo}(\what \beta))^t
  \mbox{vec}(  
\Sigma_{oo}(\what\beta))},
\end{equation}
where $\Sigma_{oo}(\beta)$ is as in \eqref{e:Sigma-beta}, $\what
\Sigma_{Y_o} = \what \Sigma_{Y_o}(t_0)$ is the sample covariance matrix of 
the vector $Y_o$, calculated from past $m$ observations $\{Y_o(t_0-k),\ 0\le k\le m-1\}$,
and $\what\beta$ is an estimate of $\beta$. 

\begin{proposition}\label{p:sigma-consistency}
Assume the conditions of Theorem \ref{t:asymptotics}{\it (ii)}. Then,
with $\what\beta = \beta_k,\ k\ge 1$, for $\what\sigma^2$ as  
in \eqref{e:sigma-hat}, we have $\what \sigma^2\stackrel{a.s.}{\to}
\sigma^2$, as $m\to\infty$. 
\end{proposition}

\noindent The proof is given in the Appendix. 

In the following section, we will address the estimation, validation
and the applications of the complete model \eqref{eq:CompMod} 
by using real and simulated data.  We conclude this section with a few
technical comments. 

\medskip
\noindent{\bf Remarks:} 
\begin{enumerate}
\item
The model in \eqref{eq:CompMod} is realistic only if $F\beta>\vec 0$.  In our
experience, the estimates $\what \beta_k$ obtained from 
real network data always satisfy $F\what \beta_k> \vec 0$. This is perhaps due to 
the careful (optimal) choice of the matrix $F$ discussed in Section \ref{sec:means}.

\item The assumption that $A_o$ is of full row--rank is natural since
for prediction purposes, one need not include in the set of observed links 
ones that are perfect linear combination of other observed links. In practice, 
such a redundant scenario can arise only in the trivial case when some nodes do 
not generate traffic.

\item The full column--rank condition on $A_oF$ is required for the identifiability of $\beta$.
If the dimension of $\beta$ is greater than the number of observed links, then the 
model parameters cannot be identified (see also Section \ref{sec:cal} and Figure 
\ref{fig:gpr}). In practice, we implement the iGLS procedure by using the 
Moore--Penroze generalized inverse.

\item Under the assumptions of Theorem \ref{t:asymptotics}, one can show that
for all $k\ge 2$,  
\begin{equation}\label{e:MSE-beta}
 \E [(\what \beta_{k} -\beta)(\what \beta_{k} - \beta)^t 1_{K}(\what \beta_{k-1})]
  = \sigma_{m}^2 \Sigma_{GLS}(\beta) + o(\sigma_m^2),
\end{equation}
as $m\to\infty$, where $K\subset \mathbb R^p$ is an arbitrary compact set containing $\beta$ 
in its interior and such that $F\wtilde \beta> \vec 0,\ \forall \wtilde\beta\in K$.
Thus, the variance of the estimators $\what\beta_k,\ k\ge 2$ obtained in 
practice is essentially asymptotically optimal.

\end{enumerate}

\section{Model Validation and Calibration}
In this section, we evaluate our model in the context of traffic
prediction. We focus on 12 representative scenarios described in Tables 
\ref{tab:I2links}, \ref{tab:cases} and the Appendix  \ref{ap:i2desc}, below.

\subsection{Performance}
Tables \ref{tab:okrig} and \ref{tab:fullmodel} provide ReMSE's \eqref{e:ReMSE} for 
the optimal baseline estimator, the ordinary kriging estimator, and our 
network--specific model (with $p=2$), respectively.

Scenarios 1--9 represent situations where the observed links share 
sufficiently many flows with the unobserved ones and thus there is enough 
information for relatively accurate prediction.  In {\em all} these cases, the
{\em network--specific model} outperforms the more naive
{\em ordinary kriging} method, with an average improvement of the ReMSE
by $0.1887$ points or $18.87 \%$. The difference is as high as $78\%$
and as low as $0.8 \%$ in favor of the network specific model. In most cases and across 
different days our model yields useful predictions with ReMSE's of about $5\%$. Note that
the optimal ReMSE's (Table \ref{tab:okrig}, baseline) are about $2.5\%$ in these scenarios.

In Scenarios 10--12, however, fewer flows are shared by the observed and
unobserved links and hence the accurate prediction is objectively more difficult.
Note that the baseline predictor (Table \ref{tab:okrig}) has over twice the 
ReMSE's in these cases as compared to cases 1--9.  In Scenarios 11 and 12, the ordinary 
kriging estimator outperforms the network specific model with differences in 
the ReMSE's between $19\%$ and $105\%$. The ReMSE's of the ordinary kriging estimator, however,
are greater than $21.32\%$. In Scenario 10, our model is superior to ordinary kriging 
for all five days, but the large ReMSE's indicate that neither approach is particularly 
useful in this case.

This initial comparison shows that the network--specific model improves significantly upon the 
naive ordinary kriging approach and comes close to the optimal ReMSE lower bound.  This is so in
the cases where the prediction problem is well--posed. Scenarios 10--12 illustrate that the accuracy of
prediction has natural limitations, inherent to the routing of the network, that none of the 
two models can overcome.

\subsection{Robustness over Time}

One apparent limitation of the network specific approach is that it 
relies on expensive flow--level data ($X_j$'s) to build the matrix $F$. 
Surprisingly, it turns out that once the matrix $F$ is obtained from flow--level 
measurements during a single day, it can be successfully used to model the link--level
traffic for many days in the future.  That is, even though the model requires the extensive 
off--line analysis of NetFlow data, once it is built, it can be readily estimated
{\em on--line} using only link--level data and used for several days before it has 
to be updated. It is remarkable that all results in Table \ref{tab:fullmodel} are based on 
a model (i.e.\ a matrix $F$) learned from Feb 19, 2009 flow--level data.  Then, the
{\em same} model was used to predict $Y_\ell$'s in all 12 scenarios for 5 different 
days. Even a month later, this model continues to outperform the ordinary kriging  
in the first nine scenarios.  Figure \ref{fig:roled} shows also that in only one of 
the $6$ scenarios therein we have an appreciable increase in the prediction ReMSE's due perhaps  
to an outdated model. These results may be attributed to the fact that the structure of the traffic means
across all flows in the network, although complex, is relatively constant, and is therefore
well--captured by the principal components involved in the matrix $F$. The model must be updated
should structural changes in the network occur.

\begin{table}
\begin{center}
\begin{tabular}{r|rrrrr}
  \hline
Scenario  & 2009-02-19 & 2009-02-18 & 2009-02-20 & 2009-02-26 & 2009-03-12 \\ 
  \hline
1 & 0.2476 & 0.2342 & 0.2363 & 0.2629 & 0.2209 \\ 
  2 & 0.0517 & 0.0461 & 0.0550 & 0.0424 & 0.0746 \\ 
  3 & 0.0514 & 0.0459 & 0.0549 & 0.0425 & 0.0750 \\ 
  4 & 0.0521 & 0.0465 & 0.0552 & 0.0427 & 0.0740 \\ 
\hline
  5 & 0.0512 & 0.0696 & 0.0658 & 0.0694 & 0.0596 \\ 
  6 & 0.2414 & 0.2651 & 0.2864 & 0.3344 & 0.2722 \\ 
  7 & 0.0468 & 0.0619 & 0.0587 & 0.0684 & 0.0462 \\ 
  8 & 0.0388 & 0.0501 & 0.0495 & 0.0564 & 0.0384 \\ 
  9 & 0.0395 & 0.0510 & 0.0504 & 0.0567 & 0.0398 \\ 
\hline
  10 & 3.6668 & 3.9110 & 3.8143 & 5.0877 & 5.5161 \\ 
  11 & 1.0322 & 1.1060 & 1.0687 & 1.4335 & 1.7803 \\ 
  12 & 0.6277 & 0.7618 & 0.6875 & 0.7792 & 0.7449 \\ 
   \hline
\end{tabular}
\end{center}
\caption{ReMSE's of the network--specific model. The matrix
  $F$ was obtained from Feb 19, 2009 NetFlow data ($X_j$'s), and 
  then used to fit the model and perform prediction based on 
  link data ($Y_l$'s) for five different days.}
\label{tab:fullmodel}
\end{table}

\begin{figure}
\centering
\includegraphics[width=.5\textwidth]{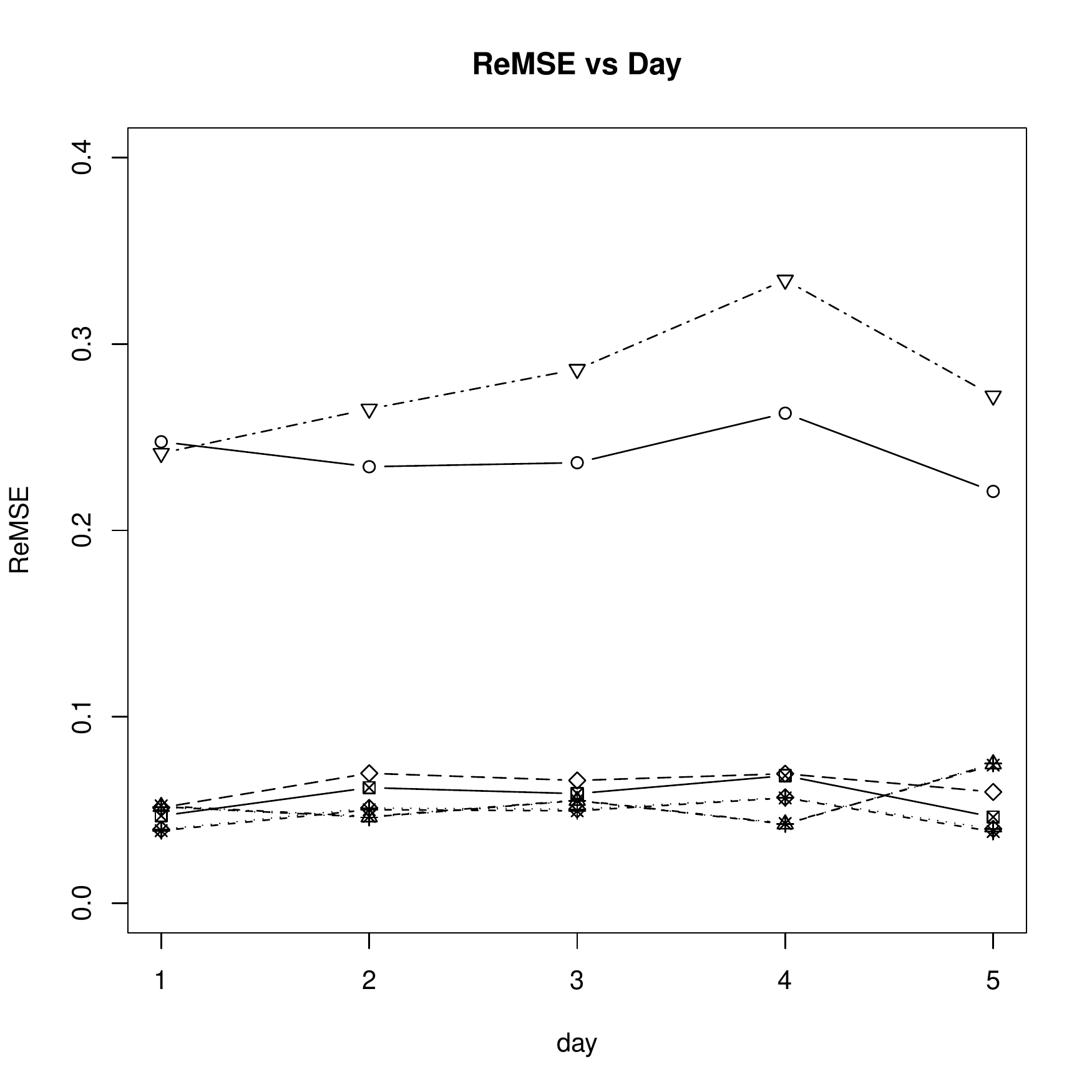}
\caption{ReMSE of network--specific model over time.  The model was
  learned on Feb 19, 2009
  (1). The matrix is then used to predict the previous day (2), the
  next day (3), a day one week later (4), and a day 4 weeks later
  (5).  Each line corresponds to one of the first six scenarios described in
Table \ref{tab:cases}.}
\label{fig:roled}
\end{figure}

\subsection{Calibration}\label{sec:cal} Applying the model to real data relies on 
the choice of several parameters, such as $p$, $\gamma$, and $m$, as described in Section
\ref{sec:network-model}.  The prediction performance is remarkably 
robust to the choice of these parameters, as discussed in detail below.

\smallskip
\noindent $\bullet$ {\it The role of $ \gamma$}: This parameter controls the 
mean/variance relationship in the model (see \eqref{eq:Fmod} and
\eqref{eq:Fcov}). We observed the relationship \eqref{e:Sj_vs_Xj}, between the 
sample means $\overline X_j(t)$ and standard deviations $S_j(t)$'s obtained 
from windows of the flow--level data. The parameter $\gamma$ was estimated by using a 
log--linear regression of $\overline X_j(t)$ versus $S_j(t)$, over 
$j,\ 1\le j\le {\cal J}$. This was done for a range of window sizes and
times $t$, and the estimates were found to be stable and $\what \gamma
\approx 3/4$, as can be clearly seen in Figure \ref{fig:gamest}.
Independently, we explored the sensitivity of
the model to the choice of $\gamma$ and found that the ReMSE's are robust to all choices  
$\gamma \in [0.5, 2]$. See Figure \ref{fig:rolegam} the effect of the
choice of $\gamma$ for several of the prediction scenarios.   Small
values of $\gamma\approx 0.5$ lead generally to slightly  
better ReMSE as compared to larger $\gamma$'s. This may be due to the fact that the
small powers $\gamma$ lead to a 'smoother' covariance matrix and hence have a regularizing effect.
In practice, however, we need not only accurate prediction but also adequate models for the variance,
in order to have reliable estimates of the prediction error. Therefore, we recommend using 
$\gamma\approx 3/4$ as inferred from the data.

\begin{figure}[t!]
\centering
\subfloat[Estimating $\gamma$ \label{fig:gamest}]{\includegraphics[width=.46\textwidth]{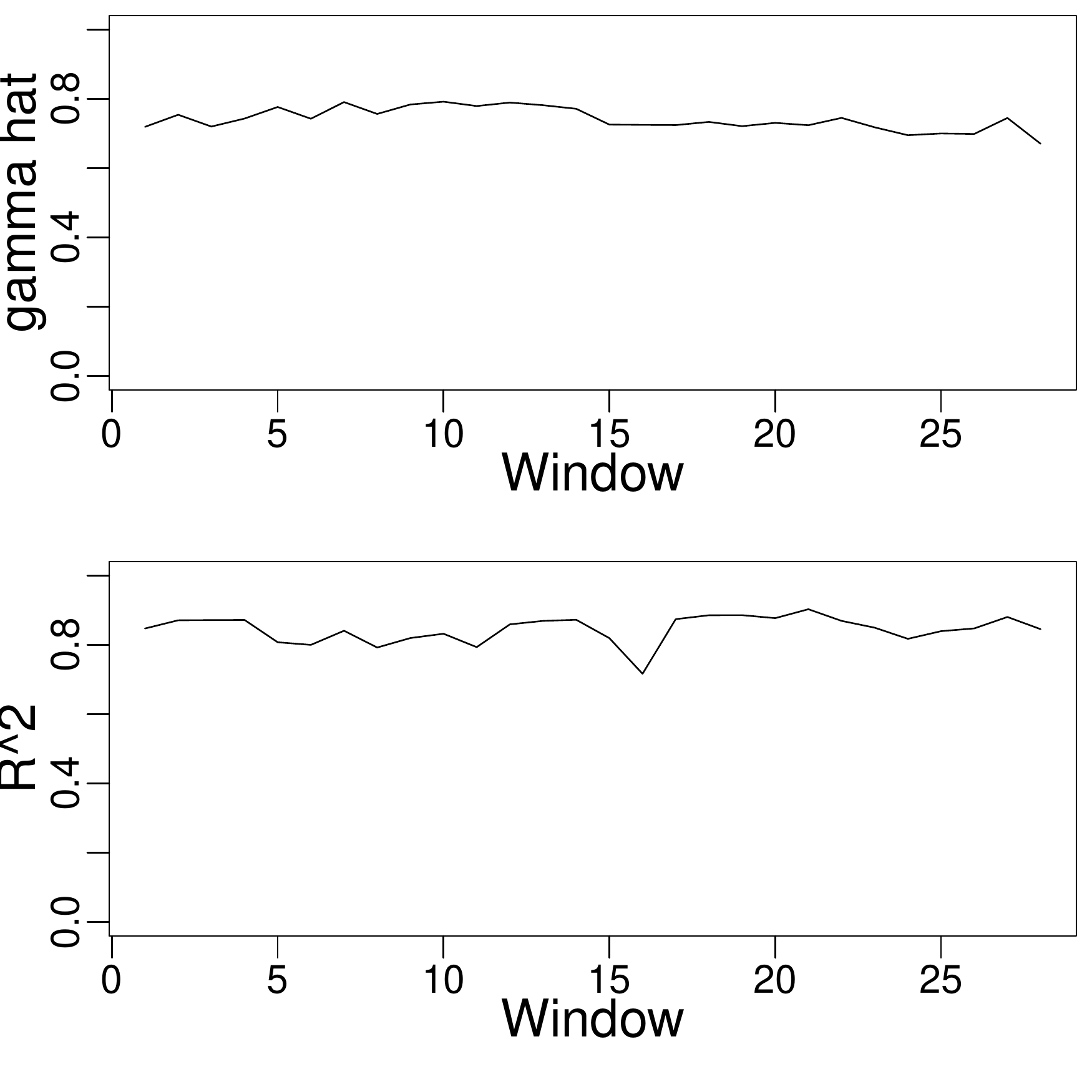}}
\subfloat[Effect of $\gamma$ on
Prediction \label{fig:rolegam}]{\includegraphics[width=.46\textwidth]{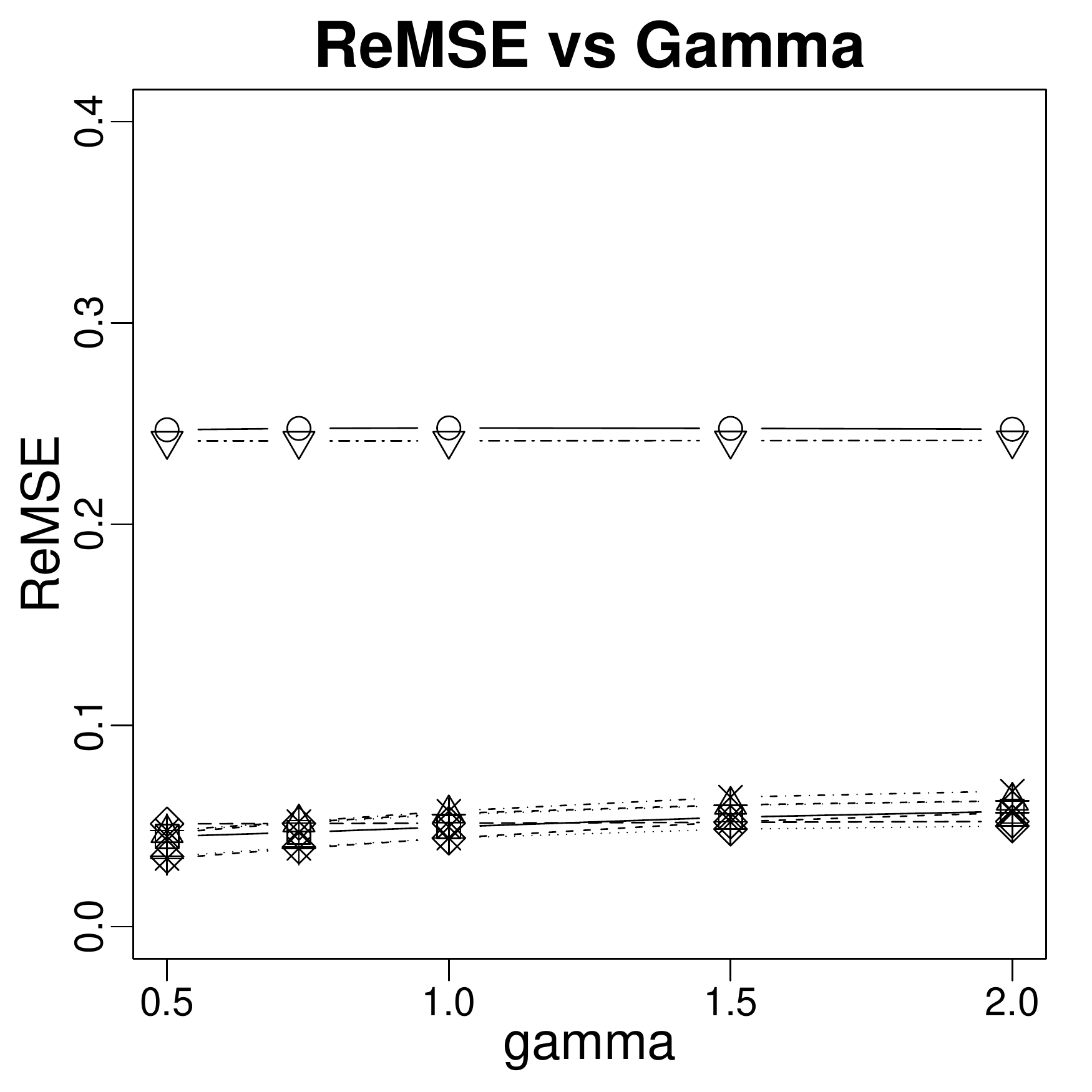}} 
\caption{{\it Left:} Estimates of
  $\gamma$ over windows of the data (top) and their $R^2$ values (bottom), obtained from 
  log--linear regression of the sample standard deviations against the means across 
  all source--destination routes. {\it Right:} ReMSE of scenarios 1-9
  for choices of $\gamma$.  Note  that the choice of this parameter
  does not significantly alter the performance of the predictor for
  the values tested.} 
\end{figure}

\smallskip
\noindent $\bullet$ {\it The role of p}: The parameter $p$ equals 
the number of principal components (columns of the matrix $F$)
used to model the traffic means in \eqref{eq:Fmod}. The prediction 
performance is robust to the choice of $p$, provided that $p$ is less than 
the number of {\em observed links} used in prediction. Figure \ref{fig:gpr} 
shows the ReMSE's for 3 prediction scenarios as a function of $p$.  
In Scenarios 6, 7, and 8 the {\em same} link is predicted via two, three, and 
seven other links, respectively (see Table \ref{tab:cases}).  
If $p$ exceeds the number of observed links, then
the parameter $\beta$ in \eqref{eq:CompMod} is not identifiable, potentially resulting 
in poor performance. This explains the peaks in the ReMSE's at $p=2,\ 4,$ and $7$ in 
Scenarios 6--8. Surprisingly, in the first two cases the ReMSE's recover as $p$ grows, 
even in the presence of non--identifiability. Similar patterns are seen in the other 9
prediction scenarios (omitted, for simplicity). The performance of the model 
remains stable for all choices of $p$ less than the number of predictors.

In light of these results, we advocate using a relatively small value
of $p$, such as $2$. While a larger value of $p$ can slightly improve prediction errors 
when many links are observed, having small value of $p$ allows one to fit the model in a 
wide variety of prediction scenarios, without sacrificing the overall performance. 
Recall also Figure \ref{fig:rolep}.

\begin{figure}[t!]
\centering
\subfloat[Role of $p$ \label{fig:gpr}]
{\includegraphics[width=.46\textwidth]{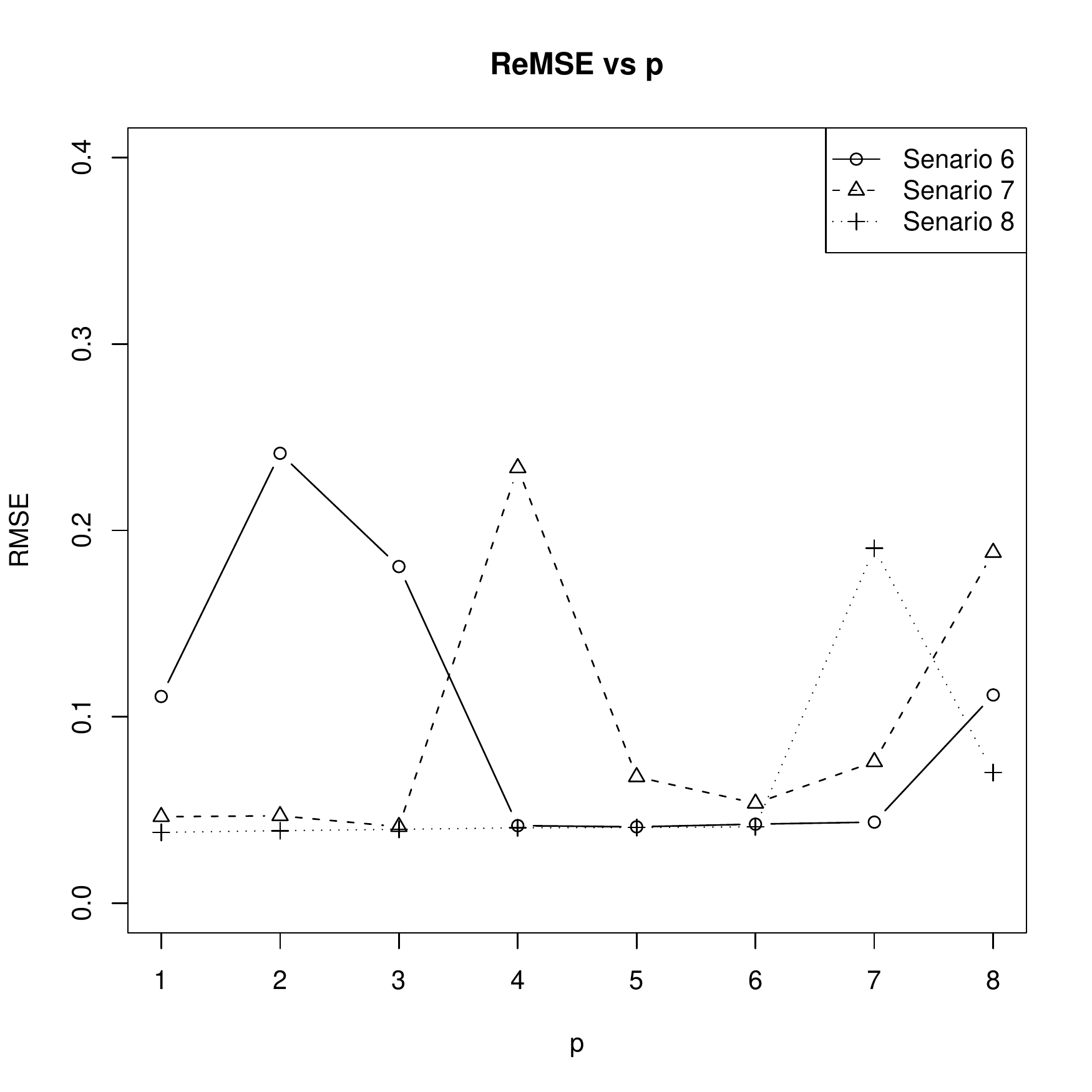}}
\subfloat[Mean Window \label{fig:rolemw}]{\includegraphics[width=.46\textwidth]{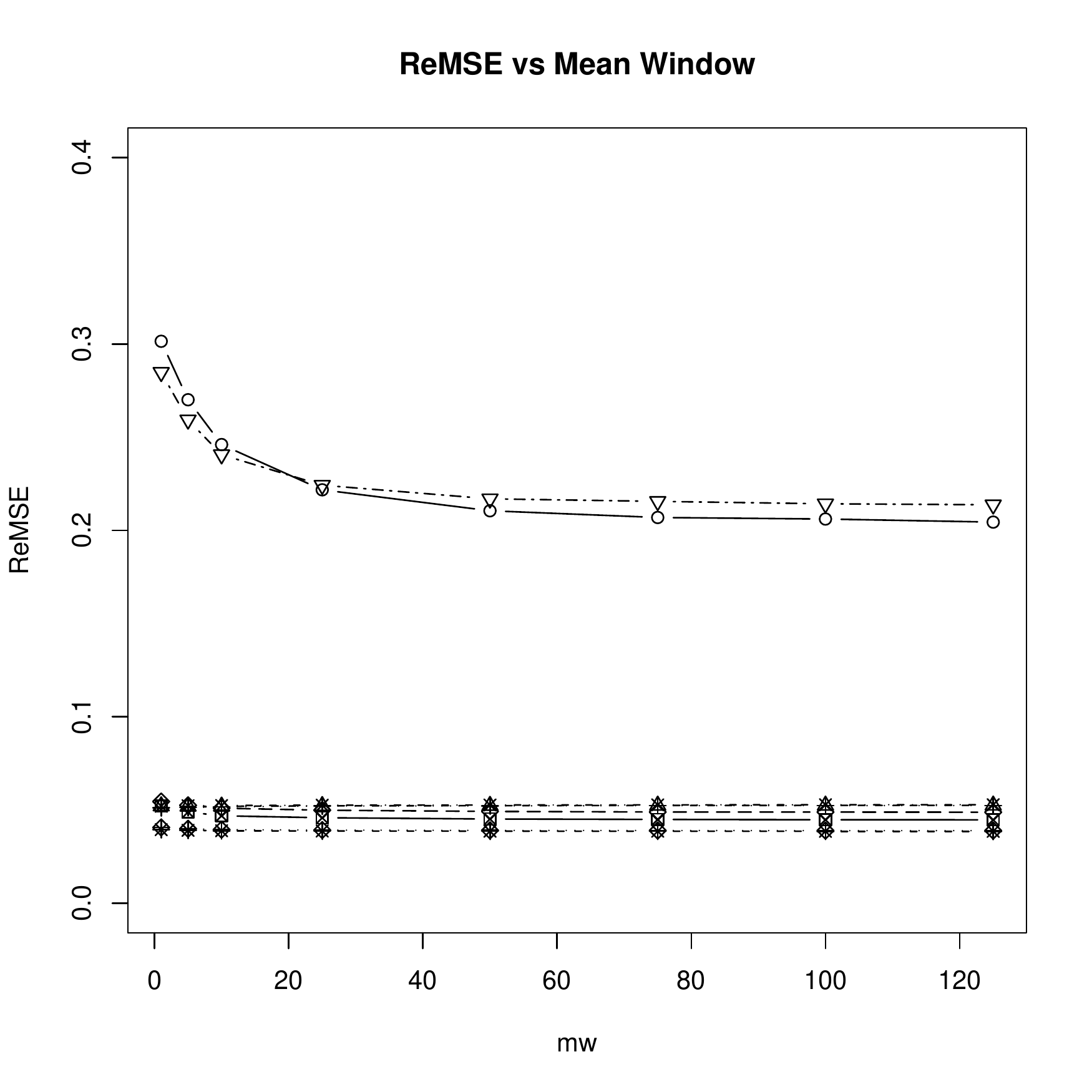}} 
%\subfloat[Role of $\gamma$]{\includegraphics[width=.46\textwidth]{exp-gams}}
\caption{{\it Left:}  ReMSE of scenarios 6, 7, and 8 as a function of
  $p$.  Performance suffers when $p$ exceeds the number of observed
  links, but is otherwise robust to the choice of $p$.  
{\it Right:} ReMSE of scenarios 1--9, as a function of window size $m$ used to
  estimate the means. We used $p=2$ in these cases.
}

\end{figure}

\medskip
\noindent $\bullet$ {\it The role of the window size $m$:}
In practice, at each time point $t$, the model \eqref{eq:CompMod} is 
estimated from a window of past $m$ data $\{Y_o(t-k),\ 0\le k \le m-1\}$.
Namely, $\beta$ is obtained by using the iGLS algorithm and $\sigma$ from 
\eqref{e:sigma-hat} (see Section \ref{s:joint_model}). Figure \ref{fig:rolemw} illustrates
the effect of the window size $m$ on the quality of  prediction. Note that in all scenarios
therein the prediction performance is rather robust to the choice of $m$, provided 
that $m \ge 10$. It is remarkable that with the exception of 2 out of the 9 shown prediction 
scenarios the model works well even when $m$ is less than $10$.

Recall that the scalar $\sigma$ does not affect the prediction and the ReMSE's in 
Figure \ref{fig:rolemw} depend only on the quality of estimation of $\beta$. 
The parameter $\sigma$ is involved in the prediction error. Our experiments with simulated data
(not shown here, for simplicity) show that the estimates of $\sigma$ are also robust to the
window size $m$.

\medskip
\noindent $\bullet$ {\it Convergence of $\widehat \beta_{iGLS}$}:  In practice, the estimates
of $\widehat \beta_{GLS}$ stabilize after a few iterations.  Here we assume
the convergence criterion $\| \widehat \beta_k-\widehat \beta_{k+1} \| <
\epsilon$, with $\epsilon=0.001$, for example.  In the
figures and tables, however, the algorithm was 
allowed to run for at least 20 iterations to be conservative.

\subsection{Model Misspecification}

Here, we apply our model \eqref{eq:CompMod} to simulated data that violates 
the assumption of stationarity. Our goal is to understand the limitations of model, when 
applied to network traffic with slowly changing trend. Network flows are simulated 
using independent fractional Gaussian noise (fGn) time series with self--similarity parameter
$H=0.8$ (see Section \ref{s:fGn}, below). This value of $H$ is typical for several real network 
flows that we examined.  

We compare the mean squared prediction errors in the stationary and non--stationary regimes.
In the stationary case, constant means are added to the simulated fGn's to produce realistic
traffic flows with constant positive means. Non--stationary traffic traces were obtained
by adding a sinusoidal trend to all simulated stationary flows. In both cases (stationary and non--
stationary), link--level data was obtained from the simulated flow--level data through the
routing equation \eqref{routing.eqn1}.

We focused on prediction Scenario 8
(see Table \ref{tab:cases}). We computed the baseline `simple kriging' predictor $\wtilde Y$ 
by using the known means and covariances of the simulated data. We also estimated our model
and used it to obtain a predictor $\widehat Y$ of the unobserved link. 

Table \ref{tab:miss} shows the resulting prediction errors as a
function of the window size used to estimate the parameter $\beta$. 

The empirical error of our estimators are comparable to the optimal MSE's 
for the baseline estimator. This is so even in the presence of non--stationarity. 
The major exception is when in the non--stationary case the window size  
becomes close to the half--period (50) of the sinusoidal trend. This limited experiment 
shows that our model adapts well and it is essentially robust to non--stationary trends provided
that relatively small window sizes $m$ are used. 

\begin{table}[ht]
\begin{center}
\begin{tabular}{r|c|rrrrrrrr}
  \hline
&  Baseline &
\multicolumn{8}{c}{Network Specific Model} \\
Window &  & 5 & 10 & 25 & 30 & 50 & 75 & 100 & 200 \\ 
\hline
Stationary& 2.61 & 2.93 & 2.88 & 2.83 & 2.82 & 2.80 & 2.78 & 2.77 & 2.75 \\ 
Non--Stationary& 2.61 & 4.15 & 4.15 & 4.50 & 4.71 & 5.42 & 6.19 & 4.99 & 4.98 \\ 
   \hline

\end{tabular}
\end{center}
\caption{Empirical mean squared errors for the baseline (simple kriging) predictor $\widetilde Y$
  and the predictor $\widehat Y$, based on our network--specific model with $p=2$.
  Time series of 20,000 observations were used.}
\label{tab:miss}
\end{table}

\section{Statistical Detection of Anomalies}
\label{s:adet}

In this section, we present an application of the above 
methodology to the case when all links on the network are observed.
In this case, for each link $\ell$, one can compare the observed
$Y_\ell(t)$ and the predicted  
$\widehat Y_\ell(t)$ traffic.  If statistically significant deviations 
are encountered, then this can serve as a flag of an anomaly or some structural
change in the network traffic.  

To illustrate and detect such differences, we use a modified exponentially weighted
moving average (EWMA) control chart on the differences $Y_{\ell}-\widehat Y_{\ell}$.
The latter have zero means and variances equal to the prediction error, which can be estimated
from \eqref{e:MSPE} and \eqref{e:sigma-hat}.

Although EWMA control charts are widely used and well--studied (see,
e.g, \cite{box:luceno:2009}), they rely on an assumption of independent or weakly dependent (in $t$)
observations. Computer network traffic is long--range dependent (LRD) and the 
usual variance formula used in the EWMA charts does not apply.  We show next how these charts 
can be adjusted to account for the presence such dependence.

\subsection{Control charts for long--range dependent data}  \label{s:fGn}
Consider the EWMA with discount factor $\phi\in(0,1)$ of the time series $\{Z_k\}_{k\geq0}$:
\begin{equation}
\widetilde Z_t:= (1-\phi)(Z_t+ \phi Z_{t-1}+\phi^2 Z_{t-2}+\cdots)
\label{eqn:ewma}
\end{equation}
Letting $\lambda= 1-\phi$, this moving average may be efficiently updated  via
$\widetilde Z_t=\lambda Z_t+(1-\lambda) \widetilde Z_{t-1}$.  
For independent $Z_t$'s, for ${\rm Var}(\widetilde Z_t) = \sigma_{\widetilde Z}^2$, we have
$ \sigma_{\widetilde Z}^2=\frac{\lambda}{2-\lambda}\sigma_Z^2.$
When the $Z_t$'s are long--range dependent, however, the latter
formula underestimates the  variance $\sigma_{\widetilde Z}^2$, which
can lead to frequent false positive alarms. 

Internet traffic traces are well known to exhibit long--range
dependence, which can be well--modeled by using fractional Gaussian noise (fGn) (see, eg
\cite{willinger:paxson:riedi:taqqu:2003-livre}). Recall that the 
zero mean Gaussian time series $\{Z_k\}$ is said to be fGn if
$Z_k=B_H(k)-B_H(k-1), k\geq 1$, where $\{B_H(t)\}_{t\geq 0}$ is a
fractional  Brownian motion with self--similarity parameter $H\in
(0,1)$.  The process $\{B_H(t)\}_{t\geq 0}$ is  self--similar and has
stationary increments. Its covariance function is 
\begin{equation*}
\E B_H(t)B_H(s)=\frac{\sigma^2}{2}{\Big(}|t|^{2H}+|s|^{2H}-|t-s|^{2H}{\Big)},\  t,s>0,
\end{equation*}
where $\sigma^2=\E B_H(1)^2=\mbox{Var}(B_H(1))$. Thus, the fGn
$\{Z_k\}$ is a stationary time series with auto--covariance 
\begin{equation}
\gamma_H(k)=\E(Z_kZ_0)=\frac{\sigma^2}{2} {\Big(} |k+1|^{2H}+|k-1|^{2H}-2k^{2H} {\Big)}.
\end{equation}
For more details, see, eg \cite{taqqu:2003-livre}. 

The following result provides an expression for the variance
$\sigma_{\widetilde Z}^2$ of the EWMA control chart corresponding to
LRD fGn data. 

\begin{proposition} \label{prop:fGnchart} For $\widetilde Z_t$ as in
  \eqref{eqn:ewma} with $Z_t$'s an fGn with self--similarity parameter
  $H\in(0,1)$ and variance $\sigma^2$, we have:  
\begin{equation}
\var (\widetilde Z_t)=  \frac{\lambda^2 \sigma^2}{C_2(H)^2}
\int_{-\infty}^{\infty}\frac{2(1-\cos(\theta))|\theta|^{-2H-1}} 
{\lambda^2-2\lambda(1-\cos\theta)+2(1- \cos \theta)} d \theta,
\label{eqn:ccfix}
\end{equation}
where $C_2(H)^2:=\pi/(H\Gamma(2H)\sin(H\pi))$.
\end{proposition}

\noindent The proof is given in the Appendix.  In practice, the
expression in \eqref{eqn:ccfix} is readily evaluated by using
numerical integration. 

\subsection{Simulated Anomalies in Observed Network Traffic}

We now present some examples of using the adjusted EWMA control chart
for real network data, applied to the differences $Y_\ell(t) - \widehat Y_\ell(t)$.
The mean is taken to be $0$ for the duration of the control chart, 
since the predictor $\widehat Y_\ell(t)$ is unbiased under the model. 
The variance of the control chart is calculated using \eqref{eqn:ccfix}, 
with $\sigma^2$ replaced by $\widehat \sigma_Y^2$ (estimated in the prediction procedure).
The Hurst LRD  parameter $H$ of the traffic is obtained by using the
wavelet--based methods described in \cite{stoev:taqqu:2003w}. 
In each of the examples, a simple mean--shift anomaly is added to one source--destination 
flow .
%with step size equal to five times the estimated mean square prediction error
%for the observed link, as calculated in the full model
%\eqref{eq:CompMod}.  \alert{inxonsistent -- how do you know which link is predicted}.
Each figure in this section shows a plot of the observed, predicted, and true traffic (top panels); 
the control chart of $|Y_{\ell}-\widehat Y_{\ell}|$ (middle panels), and an indicator 
of whether the process is identified as out of control (bottom panels). The
vertical line indicates the onset of the simulated anomaly.

\begin{figure}[h!]
\subfloat[Standard EWMA Control Chart \label{fig:CC1a}]{\includegraphics[width=.49\textwidth]{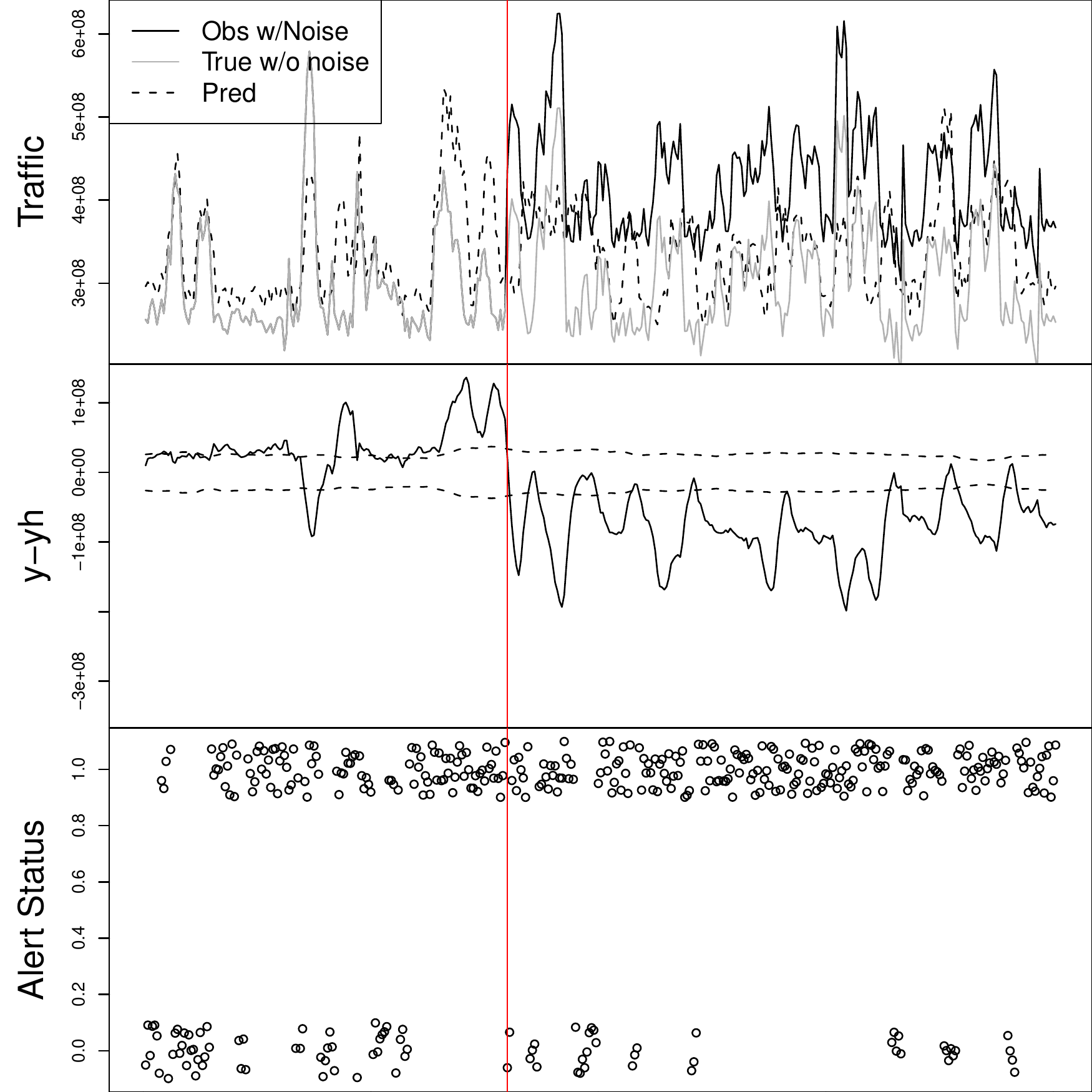}}
\subfloat[LRD--adjusted EWMA Control Chart \label{fig:CC1b}]
{\includegraphics[width=.49\textwidth]{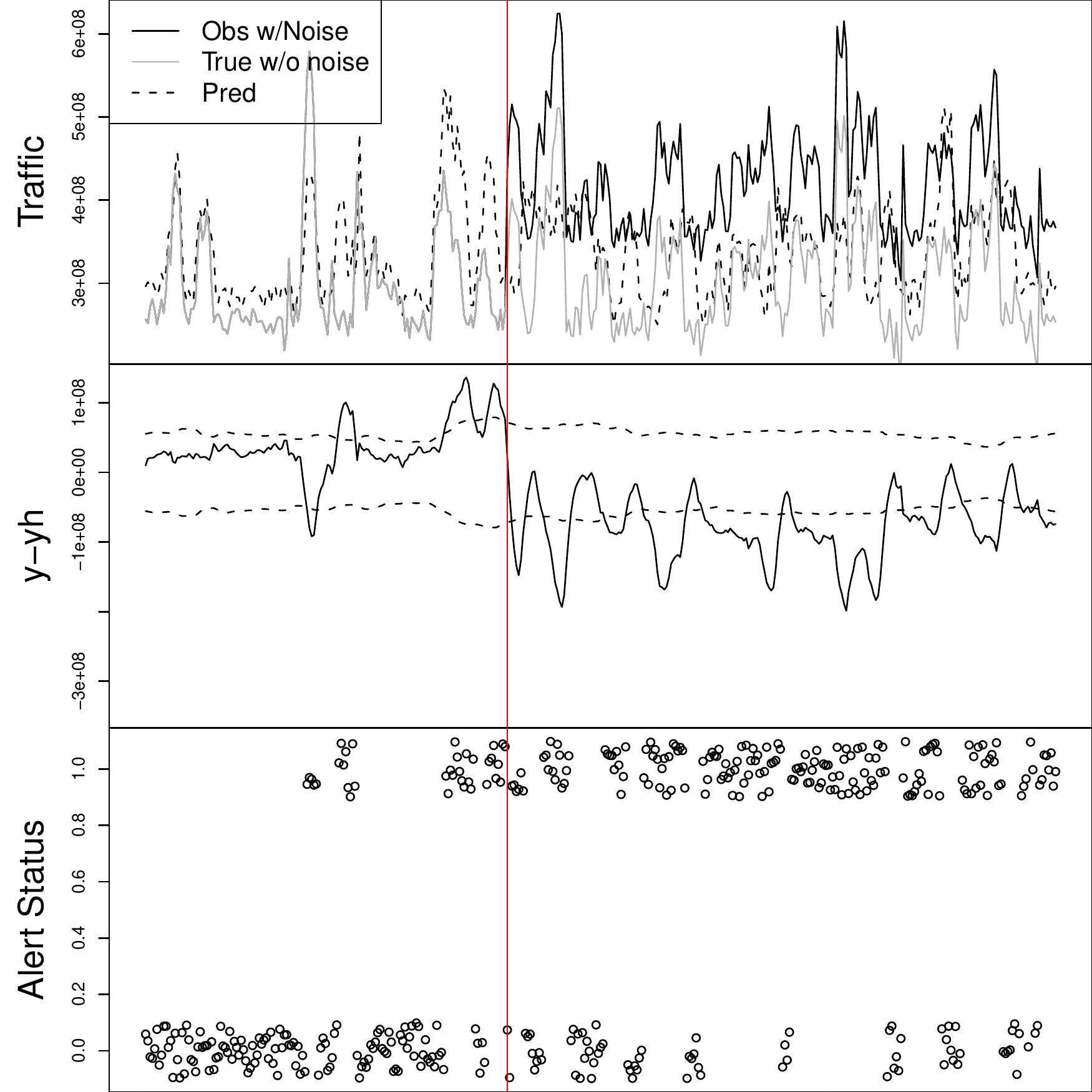}}
\caption{Performance of the standard EWMA control chart for i.i.d.\ data (left plot)
and that of the LRD--adjusted chart (right plot). }
\label{fig:CC1}
\end{figure}

Figure \ref{fig:CC1} demonstrates the importance of the LRD--adjustment for the control limits 
of the EWMA charts. Here, we examine link 13 (Kansas City to Chicago), which is predicted using 
all links sharing at least one flow with it (Scenario 8 in Table \ref{tab:cases}). 
A simulated anomaly is added to flow 20, which traverses only link 13. The standard chart 
results in far too many false positives, making it difficult to see when 
the anomaly is added.  While the adjusted chart still has several false positives 
(due to high traffic variability), the onset of the anomaly is essentially detected.

The second example shows how one can use the chart to determine
which flow is behaving anomalously.  The mean shift was added to
flow 6 (Kansas City to Atlanta), which traverses two links:
13 (Kansas City to Chicago) and 17 (Chicago-Atlanta) (see Figure
\ref{fig:scen7}). The LRD--adjusted control charts for 
Links 13 (Figure \ref{fig:CC2a}) and 17 (not shown), clearly indicate the onset of the
anomaly. The charts of the other links, not carrying the anomalous flow, 
(eg link 7 in Figure \ref{fig:CC2b}) involve just a few false alarms and detect 
no anomaly. This suggests that the flow using links 13 and 17 (that is, flow 6) 
is experiencing the anomaly.

\begin{figure}[h!]
\subfloat[Link Carrying Anomalous Flow (13)\label{fig:CC2a}]
{\includegraphics[width=.49\textwidth]{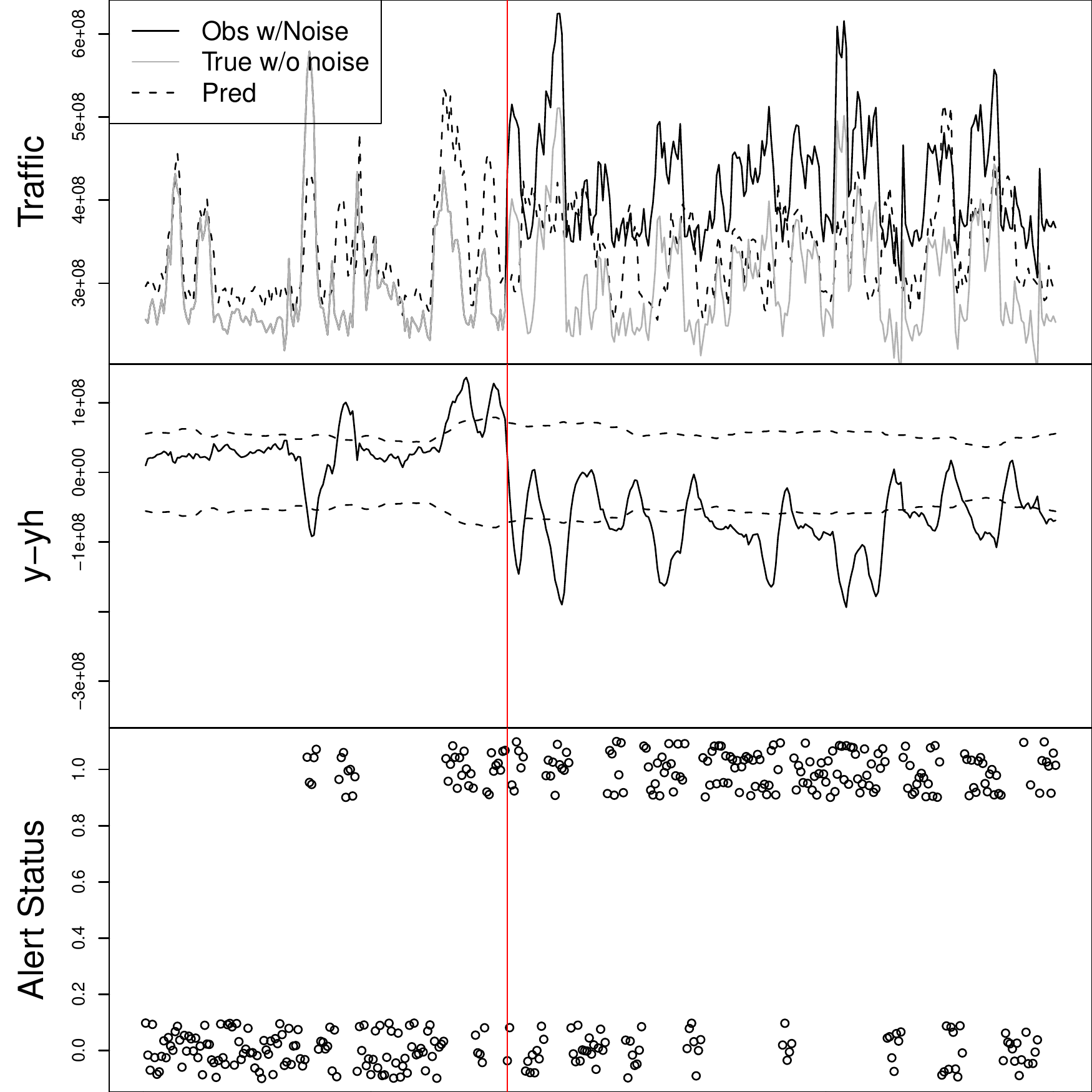}}
\subfloat[Link Not Carrying Anomalous Flow (7)\label{fig:CC2b}]
{\includegraphics[width=.49\textwidth]{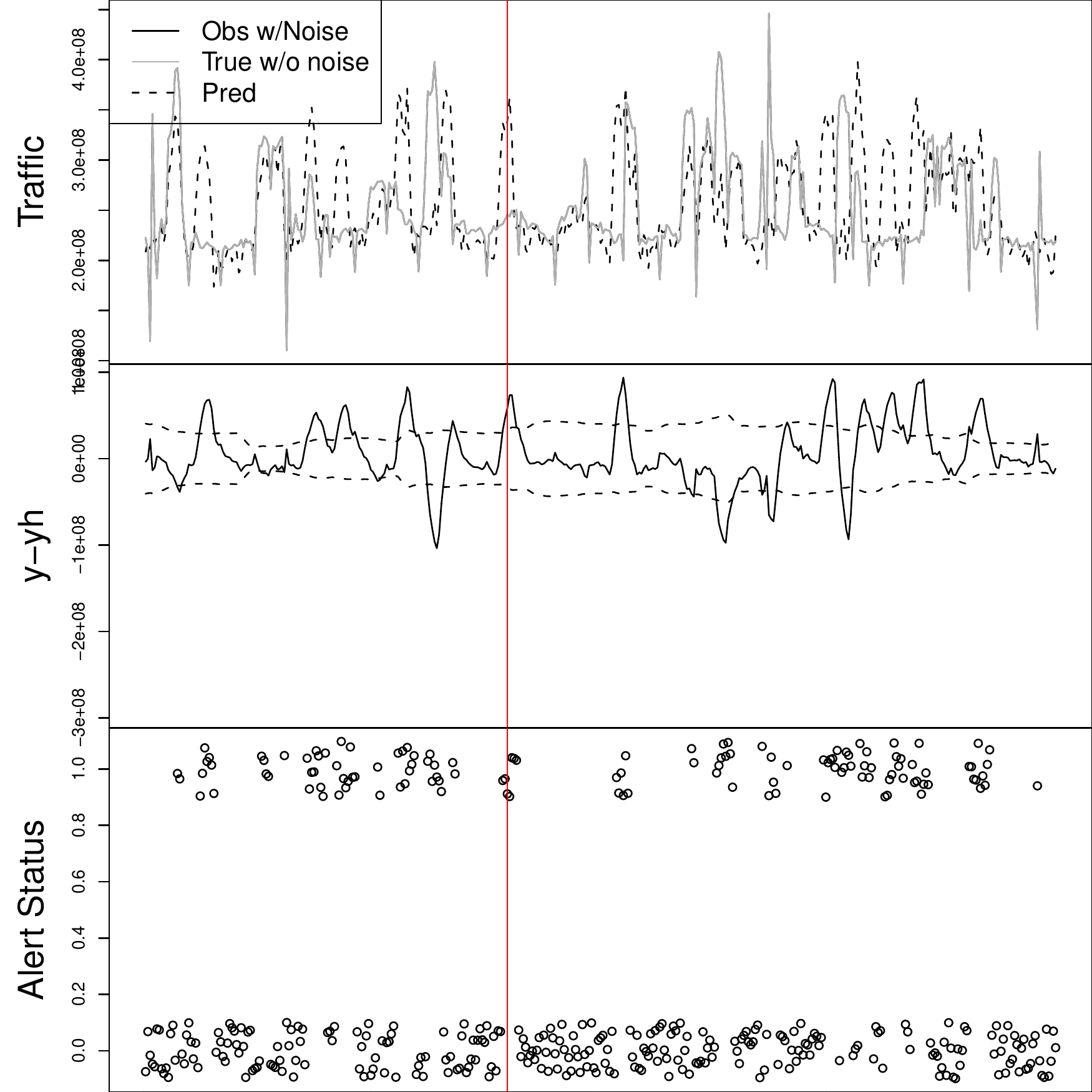}}
\caption{Detecting an anomalous multi--link flow. The simulated anomaly is
added to a two--link flow. The links carrying the flow show anomalous
behavior (left plot), while the rest behave as in the figure on right.}
\end{figure}

In the last example, we illustrate a case where the anomaly
detection is inherently more challenging.  We add a mean shift to
the relatively long flow 14 (Seattle to Atlanta), which traverses four links:
3 (Seattle to Salt Lake City), 9 (Salt Lake City to Kansas City), 13
(Kansas City to Chicago), and 17 (Chicago to Atlanta) (see Figure
\ref{fig:scen7}). In Figure \ref{fig:CC3a}, the control chart is based on 
predicting Link 17 by using all links that do not carry the anomalous flow.  
Unfortunately, these links do not provide sufficient information to predict Link 17,
hence the predictor is a relatively `smooth curve' as compared to the true traffic trace,
and the error is relatively large. Nevertheless, we can pick up the anomaly. 
The segment with false positive alerts can be explained by the presence
of "bias" in our model. That is, the model $F\beta$ is not capturing the fine 
dynamics of the means.  Indeed, if we repeat the exercise using 
simulated traffic (Figure \ref{fig:CC3b}), it shows that again the predictor is
not particularly useful i.e.\ it yields a smooth curve that  tracks only the
local means. In this situation, however, there is no bias and the control
chart accurately identifies the onset of the anomaly.

 % Figure
% \ref{fig:CC3} shows the difficult of detecting this anomaly using link
% 13.  If all the links sharing flows with link 13 are used as
% predictors, then $\widehat Y$ is influenced by the presences of the
% anomaly on some of these links, and tracks the anomalous $Y_u$ so
% closely that the process never seems out of control (Figure
% \ref{fig:CC3}a).  On the other hand, if links 3,9, and 17 are removed
% form the set of links used to predict link 13, the resulting $\widehat Y$
% measurements, although not influenced by the anomaly, are insufficient
% for an accurate prediction, making the entire chart appear out of
% control. 

\begin{figure}[t!]
\subfloat[Real Traffic \label{fig:CC3a}]{\includegraphics[width=.49\textwidth]{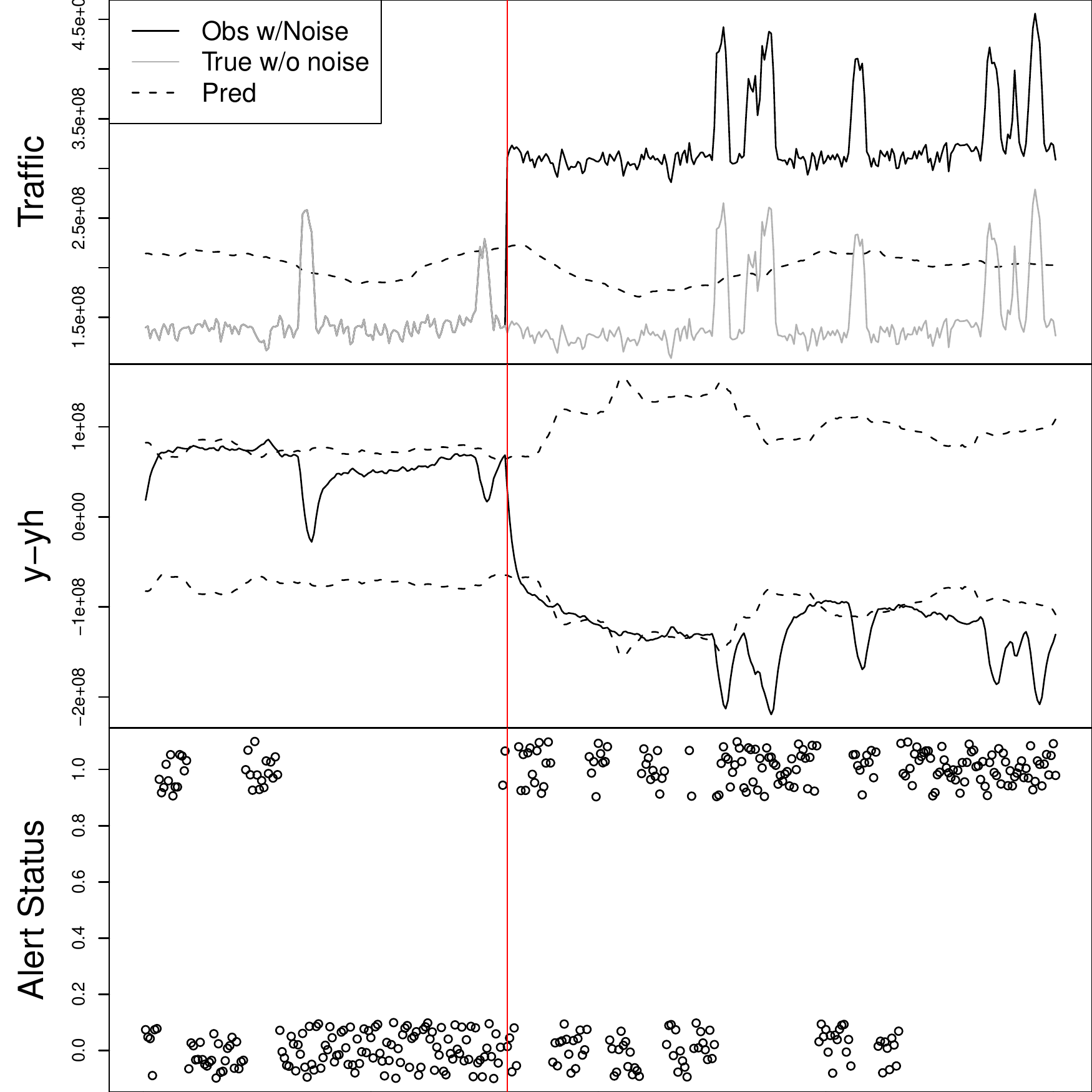}}
\subfloat[Simulated Traffic\label{fig:CC3b}]{\includegraphics[width=.49\textwidth]{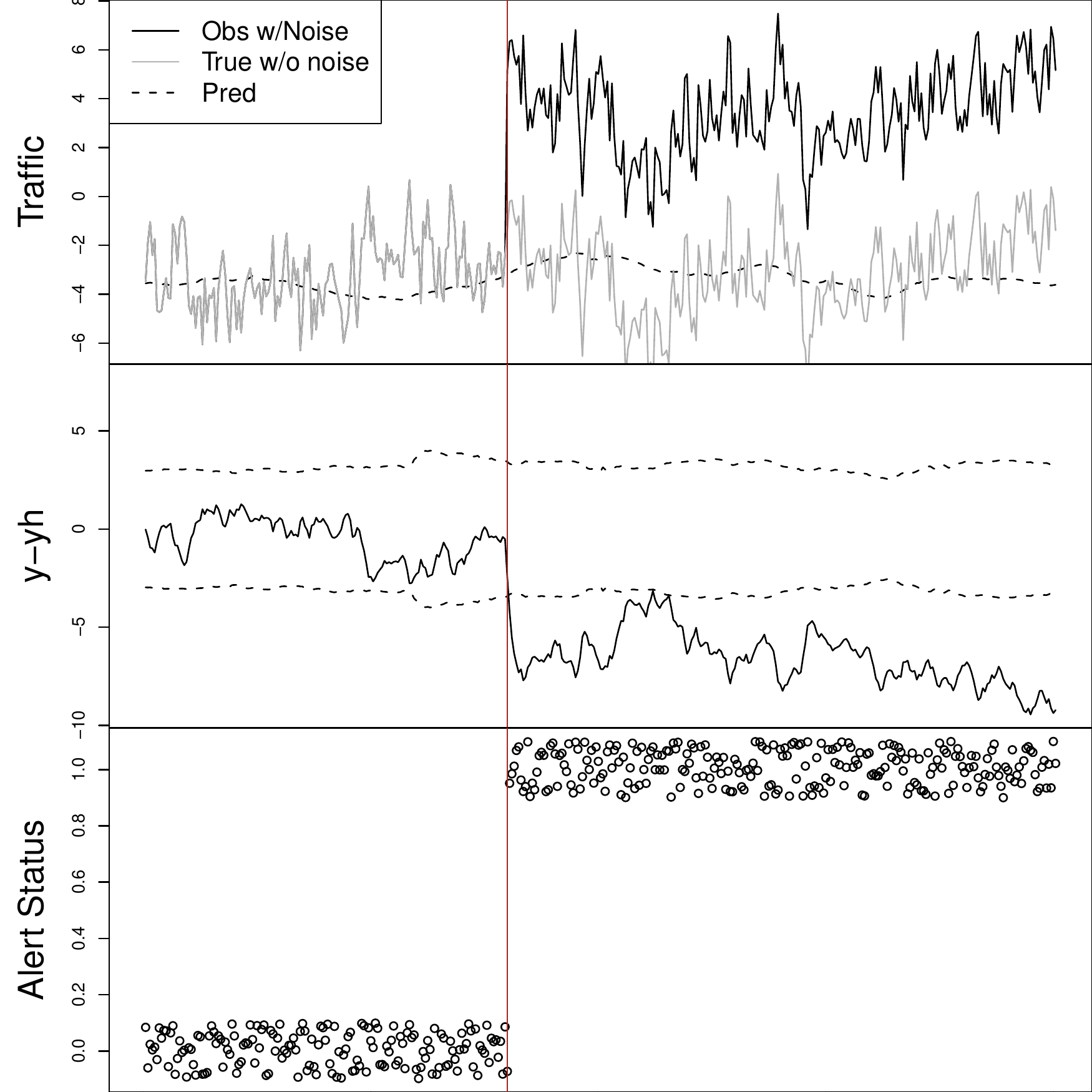}}
\caption{An anomaly is added to flow 14 (Seattle to Atlanta), and a
  control chart is constructed on the Chicago--Atlanta Link (17), where all 'non--anomalous'
  links are used in the prediction. These links, however, do not provide enough information 
  and the predictor is a relatively smooth curve.}
\end{figure}

% \section{Discussion and Future Work}

% In practice, there are some challenges and potential limitations to
% the procedure outlined above.  First and foremost is the fact that the
% methods for deriving the sampled $X$ data are imperfect and can lead
% to bias. 

%\appendix

\section{Appendix}

\subsection{Internet2 Network Description and Prediction
  Scenarios} \label{ap:i2desc} 

\begin{table} %% Note to Joel: apparently,we are not allowed to use '[b!]' qualifiers here because
              %% the stupid style file interprets the table as a figure and messes-up the labeling.
\begin{tabular}{|lll|}
\hline
\textbf{Link ID} & \textbf{Source $\to$ Destination} & \textbf{Capacity} \\
\hline
1,2 & Los Angeles $\to$ Seattle & 10 Gb/s  \\
3,4 & Seattle $\to$ Salt Lake City & 10 Gb/s \\
5,6 & Los Angeles $\to$ Salt Lake City & 10 Gb/s \\
7,8& Los Angeles $\to$ Houston  & 10 Gb/s \\
9, 10 & Salt Lake City $\to$ Kansas City  & 10 Gb/s \\ 
11, 12 & Kansas City $\to$ Houston & 10 Gb/s\\ 
13, 14 & Kansas City $\to$ Chicago & 20 Gb/s\\
15, 16 & Houston $\to$ Atlanta & 10 Gb/s  \\
17, 18 & Chicago $\to$ Atlanta & 10 Gb/s  \\
19, 20 & Chicago $\to$ New York & 10 Gb/s  \\
21, 22 & Chicago $\to$ Washington & 10 Gb/s \\
23, 23 & Atlanta $\to$ Washington & 10 Gb/s \\
25, 26 & Washington $\to$ New York & 20 Gb/s \\
\hline
\end{tabular}
\caption{\label{tab:I2links} ID's of the 26 links of the Internet2
  backbone.  Odd Link ID's correspond to the forward
  and the even to the reverse; i.e. Link 15 is the
  Houston to Atlanta link and Link 16 is the Atlanta to Houston link.
}
\end{table}

This appendix provides details concerning the Internet2 network, as
well as the 12 prediction scenarios that were investigated earlier.

The Internet2 network consists of 26 unidirectional links.  In order
to simplify notation, each link was assigned an id number.  Table
\ref{tab:I2links} provides the mapping from the link id numbers 
to the source and destination of each link, as well as the link
capacities at the time of data collection.  At the
time of the data collection, all links had a 10 Gb/s capacity, with
the exception of 4 links:  Chicago to Kansas City, Kansas City to Chicago,
New York to Washington, and Washington to New York. These four links
actually were comprised of two 10 Gb/s capacity cables, for a
total capacity of 20 Gb/s.  Similar upgrades have been made to other
links since the data presented in this work was collected.

In the preceding analysis, estimators are compared for 12 different
prediction scenarios.  These scenarios are summarized in Table
\ref{tab:cases}.  In the scenarios, a total of three links are treated
as unobserved, and a subset of the remaining links are used as
predictors.  The choice of these links was not entirely arbitrary.
Recall that the traffic on any single link is equal to the sum of the
traffic of the \emph{flows} that utilize the link.  The three
unobserved links were chosen to represent a range in the utilization
level of the links (in terms of number of flows).  Link 7, predicted
in scenarios 1-4, is used by a {\it medium} number of flows.  Link 13,
predicted in scenarios 5-9, is used by 14 flows, the {\it most} of any
link.  Finally, link 19 is utilized by a {\it small} number of flows.   

The links used as predictors were also chosen based on the number of
source/destination flows shared with the unobserved link. Namely, for
each unobserved link, the predictors were selected so that they share at
least one source/destination flow with the unobserved link.  The
exceptions are scenarios 4, 9, and 12, which involve the same set of
predictors.  In these three scenarios, our goal is to better
understand the effect of utilization on the performance of the model.  

\begin{table}
\begin{tabular}{|c|c|l|}
\hline
\textbf{Case} & \textbf{Predicted} & \textbf{Observed Links}\\
\hline
1 & 7 & 2,12 \\
2&7&2,12,13,15 \\
3& 7& 2,12,13,15,23,25 \\
4 & 7 & 2,3,9,12,15,21,23,25 \\
\hline
5 & 13 & 3,7 \\
6&13& 3,9 \\
7&13& 3,9,12\\
8&13& 3,7,9,12,17,19,21\\
9&13& 2,3,9,12,15,21,23,25 \\
\hline
10& 19 & 3,9 \\
11& 19 & 3,9,13 \\
12 & 19 &  2,3,9,12,15,21,23,25 \\
\hline
\end{tabular}
\caption{\label{tab:cases} Description of 12 Scenarios (cases) used to evaluate the model.  The choice
  of predictors is based on the number of shared traffic flows. The link id's are given in Table \ref{tab:I2links}.}
\end{table}

\subsection{Constructing Sampled Flow--Level Data}\label{sec:map}

In this section, we briefly describe how the $X$ traffic series were
created from NetFlow data. Rather than provide a complete technical
description, we seek to provide intuition for how the data were
constructed.
 NetFlow records consist of individual records, each
 that group a sequence of packets sharing
several common characteristics 
including: source and destination IP address, protocol, source and
destination port, and type of service.  These records
also include
as the number of packets, number of bytes, the start, and end time
of the group.  The input and output interfaces were of particular importance
in the mapping procedure, since  they allow us to identify, for each
record, the physical address 
from where that stream of packets arrived at the router, and the
physical address of the next
step on the network.  In particular, it allows one to determine
whether or not the series of packets arrived from the Internet2
backbone and whether it was sent out on the backbone, or to a
different network.  In the first pass over the data, we create a careful
mapping from each (source IP, destination IP) pair to a (Source
Router, Destination Router) pair.  Then we pass through the data a
second  time, using our mapping to assign each observed flow to Source
Router 
and Destination Router, and thereby assign each observed packet to one
of the 72 source/destination flows present on the backbone.

\subsection{On the implementation of simple and ordinary kriging}
\label{ap:okrig}

The baseline estimator in Section \ref{sec:general} is essentially the
{\it simple kriging} predictor, which assumes knowledge of the mean 
and covariance of $  Y$ ($\mu_Y(t)$ and $\Sigma_Y(t)$). In practice, we
estimate these quantities from moving windows of past data: 
$\mu_Y(t)\approx \what \mu_Y(t) := \frac{1}{m} \sum_{j=1}^{m} Y(t-j)$  and 
$$ 
\Sigma_Y(t) \approx \what \Sigma_Y(t):= \frac{1}{m-1}
  \sum_{j=1}^m (Y(t-j) - \what\mu_Y(t)) (Y(t-j) - \what\mu_Y(t))^t.
$$  

The {\it ordinary kriging} methodology is used in our first solution
of the prediction problem 
in Section \ref{sec:general}. In this case, $\Sigma_Y$ is modeled by
$\sigma_X^2 A A^t$, where the scale $\sigma_X$ is unknown. 
The means $\E Y_\ell = \mu_Y$ are unknown but assumed to be constant
across the links $1\le \ell\le L$.  
Here $\sigma_X$ and $\mu_Y$ are allowed to vary slowly with time
$t$. In contract, to the baseline estimator,  
we can no longer use $\what \mu_Y(t)$ and $\what\Sigma_{Y}(t)$ above
since only some links are observed.  
Under these assumptions, the least squares optimal linear predictor of
a link $Y_u(t)$ from  $Y_o(t)$ becomes 
$$
 \what Y_u (t) = \Lambda Y_o(t),\ \mbox{ where }\ \Lambda 
 = \left( \begin{array}{ll} \Gamma_{oo}  &\vec 1\\  
   \vec 1^t & 0 \end{array} \right)^{-1}\vec \gamma_{ou}, 
$$
where $\Gamma_{oo} = ( \E (Y_{\ell_i} - Y_{\ell_j})^2
)_{\ell_i,\ell_j \in {\cal O}}$ is a matrix of the variograms for the  
set of observed links ${\cal O}$ and $\vec \gamma_{ou} = (\E
(Y_{u} - Y_{\ell})^2)_{\ell\in {\cal O}}$ is a vector of
cross--variograms between the unobserved link and observed links.  For
more details, see 
\cite{cressie:1993}.  The resulting ordinary kriging coefficients are 
such that $\vec 1^t \Lambda = 1$ so that the predictor is unbiased. In
our application, we calculated the variograms by   
estimating the unknown parameter $\sigma_X$ from a window of past
data from the {\em observed links} $Y_o(t)$.  Namely, since  
$\Sigma_{Y_o} = \sigma_X^2 A_o A_o^t,$ we obtain the linear regression estimate
\begin{equation}\label{e:sigma_X(t)}
\what \sigma_X^2 = [\mbox{vec}(A_o A_o^t)^t \mbox{vec}(A_o
A_o^t)]^{-1}
\mbox{vec}(\what \Sigma_{Y_o})^t \mbox{vec}(A_o A_o^t),
\end{equation}
where $\mbox{vec}(B)$ stands for the vectorized matrix $B$.  The
estimator $\what \sigma_X^2$ corresponds to minimizing 
$\|\mbox{vec}( \what \Sigma_{Y_o}) - \sigma^2\mbox{vec}(
A_oA_o^t)\|$ with respect to $\sigma^2$.  

\subsection{Proofs}\label{sec:proofs}

\begin{proof}[Proof of Proposition \ref{p:PCA}]
Let $W = {\rm span}\{\vec f_1,\cdots,\vec f_p\}$, where
$\{\vec{f}_1,\cdots,\vec{f}_{m}\}$ is 
an orthonormal basis of $\bbR^{m}$. Observe that, for all $1\le k\le n$:
$$
 \|x(k)-P_W(x(k))\|^2=\sum_{j=p+1}^{m} \langle x(k),\vec{f}_j \rangle^2=
\sum_{j=p+1}^{m} \vec{f}_j^t(x(k)x(k)^t)\vec{f}_j.
$$
Now, by summing over $w$, we have that the right--hand side of
\eqref{e:p:PCA} equals: 
$$
 \sum_{j=p+1}^{m} \vec{f}_j^t {\Big(} \sum_{k=1}^n x(k)x(k)^t {\Big)} \vec{f}_j
 =\sum_{j=p+1}^{m}  \vec{f}_j^t B \vec{f}_j.
$$
Clearly, the last sum is minimized when 
${\rm span}\{\vec f_{p+1},\cdots,\vec f_m\} = (W^*)^{\perp}\equiv{\rm
  span}\{\vec b_{p+1},\cdots,\vec b_m\}$. 
In this case, this sum equals $\sum_{j=p+1}^m \lambda_j$.
\end{proof}

\begin{proof}[Proof of Proposition \ref{p:GLS}] Since $F\beta>\vec 0$, the matrix 
${\rm diag}(|F\beta|^{2\gamma})$ is positive definite. Thus, the fact
that $A_o$ is of full {\em row--rank}, implies that 
the square matrix  $A_o {\rm diag}(|F \beta|^{2\gamma})A_o^t$ is of
full row--rank and hence invertible. 
This proves {\it (i)}.

To show {\it (ii)}, let $x\in \mathbb R^p$ and suppose that $x^t
(A_oF)^t G(\beta) A_oF x = 0$.  Thus, 
for the vector $y =A_oF x$, we have $y^t G(\beta) y =0$.  This, since
$G(\beta)$ is positive definite,  
implies that $y = A_oF x=  \vec 0$, which in turn yields  $x = \vec
0$, because $A_oF$ has a trivial null--space. 
We have thus shown that $(A_oF)^t G(\beta) A_oF$ is a positive definite matrix.
\end{proof}

\begin{proof}[Proof of Theorem \ref{t:asymptotics}]
Note that
$$
\overline Y_o(t_0) = A_oF\beta + \sigma A_o {\rm diag}(|F\beta|^\gamma) \overline Z(t_0),
$$
where $\overline Z(t_0)= \frac{1}{m} \sum_{t=t_0-m+1}^{t_0} Z(t).$ 
Since $\rho(\tau)\to 0,\ \tau\to\infty$, the Maruyama's Theorem implies that the Gaussian process
$Z=\{Z(t)\}_{t\in \mathbb Z}$ is mixing.  Therefore, $\overline Z(t_0)
\stackrel{a.s.}{\to} 0,$ and hence  
$\overline Y_o(t_0) \stackrel{a.s.}{\to} A_oF\beta,$ as $m\to\infty$.

Note that the OLS estimator $\what \beta_1$ is well--defined since $A_oF$ is of full column rank.  
Observe also that, since $\overline Y_o(t_0)\stackrel{a.s.}{\to} A_oF\beta$, we have
$$
\what \beta_1 - \beta = [(A_oF)^t A_oF]^{-1} (A_oF)^t (\overline Y_o(t_0) - A_o F\beta) 
 \stackrel{{\rm a.s.}}{\longrightarrow} 0,\ \mbox{ as }m\to\infty.
$$
We proceed by induction.  Let $k\ge 2$ and $\what \beta_{k-1} \stackrel{a.s.}{\to} \beta,\ m\to\infty$.
Then, since $F\beta >\vec 0$, by continuity, $F\what \beta_{k-1} >
\vec 0$, almost surely, as $m\to\infty$, and 
$\what \beta_k$ is well--defined, as $w\to\infty$ (Proposition \ref{p:GLS}).
Further, for all $k\ge 2$, by \eqref{e:beta-k}, 
\begin{equation}\label{e:beta-via-C}
\what \beta_k = C(\what \beta_{k-1})\overline Y_o(t_0), 
\end{equation}
with the matrix
\begin{equation}\label{e:C(beta)}
 C(\wtilde \beta):=[(A_oF)^t G(\wtilde \beta) A_oF]^{-1} (A_oF)^t G(\wtilde \beta).
\end{equation}
Note that $C(\wtilde \beta)$ is well--defined and continuous for
$F\wtilde \beta>\vec 0$. 
Since, also $C(\beta) A_oF\beta = \beta$, the convergences 
$\overline Y_o(t_0) \stackrel{a.s.}{\to} A_o F\beta$, and $\what
\beta_{k-1} \stackrel{a.s.}{\to} \beta,$  
imply that $\what \beta_k \stackrel{a.s.}{\to} \beta$, as
$m\to\infty$.  We have thus shown parts {\it (i)} and 
{\it (ii)}.

\smallskip
We shall now prove {\it (iii)}. 
As in \eqref{e:beta-via-C}, for all $k\ge 2$,  we have
$$
\what \beta_{k}  = C(\what\beta_{k-1})\overline Y_o(t_0)\ \mbox{ and also }\ 
\what\beta_{GLS} = C(\beta) \overline Y_o(t_0).
$$
Note also that $C(\beta) A_oF\beta = \beta= C(\what\beta_k) A_oF\beta$, and therefore,
\begin{equation}\label{e:beta-k-beta-gls}
\what \beta_k - \what \beta_{GLS}  
 = (C(\what\beta_{k-1}) - C(\beta) ) ( \overline Y_o(t_0) - A_oF\beta).
\end{equation}
Now, by \eqref{e:epsilon-Y-bar} and \eqref{e:G(beta)}, we have 
\begin{equation}\label{e:Y-bar-var}
 {\rm Var}(\overline Y_o(t_0)) = \sigma_m^2 G(\beta)^{-1}.
\end{equation}
Thus, Relation \eqref{e:beta-k-beta-gls}, the convergence
$\what\beta_{k-1}\stackrel{a.s.}{\to}\beta,\ m\to\infty$, 
and the continuity of $C(\cdot)$ imply that $\what \beta_k = \what
\beta_{GLS} = o_P(\sigma_m),\ m\to\infty$. 
Note also that by \eqref{e:C(beta)} and \eqref{e:Y-bar-var}, we readily have 
${\rm Var}(\what \beta_{GLS}) = \sigma_m^2 \Sigma_{GLS}(\beta)$.
\end{proof}

\begin{proof}[Proof of Proposition \ref{p:sigma-consistency}]
As in the proof of Theorem \ref{t:asymptotics}, the Maruyama's theorem implies that
$\what \Sigma_{Y_o} \stackrel{a.s.}{\to} \Sigma_{Y_o},$ as $m\to\infty.$
By Theorem \ref{t:asymptotics}  
{\it (ii)} we also have $\what\beta_k \stackrel{a.s.}{\to} \beta,$ as $m\to\infty$. 
Note that the right--hand side of \eqref{e:sigma-hat} is a continuous
function of $\what\beta$ and $\what \Sigma_{Y_o}$, 
which by \eqref{eq:CompMod}, equals $\sigma^2$ when $\what \beta$ and
$\what \Sigma_{Y_o}$ are replaced by $\beta$ and 
$\Sigma_{Y_o}$, respectively.  This implies the strong consistency of $\what \sigma^2$.
\end{proof}

\begin{proof}[Proof of Proposition \ref{prop:fGnchart}]
The variance of $\widetilde Z_t$ is then given by:
$$
  \E [\widetilde Z_t^2] =  (1-\phi)^2 \sum_{j=0}^{\infty} \sum_{k=0}^{\infty}
  \phi^{j+k}\gamma(j-k) =  (1-\phi)^2 \sum_{j=0}^{\infty} \sum_{k=0}^{\infty}
  \phi^{j+k} \int_{-\pi}^\pi e^{i\theta(k-j)}f(\theta)d\theta,
$$ 
where $f(\theta)$ stands for the spectral density of $\{Z_k\}$. By
using the expression of $f$ given in equation 9.12 on p.34 of \cite{taqqu:2003-livre}
, we obtain that $\E[\widetilde Z_t^2]$ equals
\begin{eqnarray*}
  & & \lambda^2\int_{-\pi}^\pi {\Big|}\sum_{k=0}^{\infty}\phi^ke^{i\theta k}{\Big|}^2 f(\theta) d\theta 
      = \lambda ^2C_2(H)^{-2}
      \int_{-\infty}^{\infty}\frac{2(1-\cos(\theta))|\theta|^{-(2H+1)}}{\phi^2+1-2\phi
        \cos \theta} d \theta \\ 
& &  \ \ =  \lambda^2 C_2(H)^{-2}
\int_{-\infty}^{\infty}\frac{2(1-\cos(\theta))|\theta|^{-(2H+1)}}{\lambda^2-2\lambda(1-\cos\theta)+2(1-
  \cos \theta)} d \theta. 
\end{eqnarray*}
\end{proof}

\bibstyle{agsm}
%\bibliography{mst-bibfile}

\end{document}